\documentclass[conference]{IEEEtran}
\IEEEoverridecommandlockouts
\usepackage{
amsmath,cite,bm,amsfonts,amssymb,
graphicx,color,mathtools,setspace,
comment,multirow,xfrac,epstopdf,
footnote,multirow,lscape,float,
caption,subcaption,amsthm,alltt,placeins,amsthm}

\newtheorem{theorem}{Theorem}
\newtheorem*{result}{Result 1}
\newtheorem*{corollary}{Corollary}

\newcommand{\raisecapt}{\vspace{-0.1cm}}
\newcommand{\myVM}[3]{\mathbf{#1}_{\mathrm{#2}}^{#3}} 
\newcommand{\myVMIndex}[4]{\mathbf{#1}_{\mathrm{#2},#3}^{#4}} 
\newcommand{\Norm}[3]{\left\lVert #1 \right\rVert_{#2}^{#3}} 
\newcommand{\tr}[3]{\mathrm{tr}\left\{\mathbf{#1}_{\mathrm{#2}}^{#3}\right\}}
\newcommand{\eq}[1]{(#1)}
\newcommand{\Diag}[1]{\mathrm{diag}\{#1\}}
\newcommand{\Abs}[2]{\left\lvert #1 \right\rvert^{#2}}
\newcommand{\Exp}[1]{\mathbb{E}\left\{ #1 \right\}}
\newcommand{\Var}[1]{\mathrm{Var}\left\{ #1 \right\}}
\newcommand{\LagP}[1]{L_{1/2}\left( #1 \right)}

\newcommand{\betad}[0]{\beta_{\mathrm{d}}}
\newcommand{\betabr}[0]{\beta_{\mathrm{br}}}
\newcommand{\betaru}[0]{\beta_{\mathrm{ru}}}
\newcommand{\etad}[0]{\eta_{\mathrm{d}}}
\newcommand{\zetad}[0]{\zeta_{\mathrm{d}}}
\newcommand{\etaru}[0]{\eta_{\mathrm{ru}}}
\newcommand{\zetaru}[0]{\zeta_{\mathrm{ru}}}

\newcommand{\kd}[0]{\kappa_{\mathrm{d}}}
\newcommand{\kru}[0]{\kappa_{\mathrm{ru}}}
\newcommand{\rhoruij}[1]{\lvert \rho_{ik} \rvert^{#1}}

\newcommand{\ad}[1]{\myVM{a}{d}{#1}}
\newcommand{\hdNLOS}[1]{\myVM{R}{d}{1/2} \myVM{u}{d}{#1}}
\newcommand{\hdtil}[1]{\myVM{\tilde{h}}{d}{#1}}

\newcommand{\aru}[1]{\myVM{a}{ru}{#1}}

\newcommand{\hrutil}[1]{\myVM{\tilde{h}}{ru}{#1}}

\begin{document}

\bstctlcite{IEEEexample:BSTcontrol} 

\title{Optimal SNR Analysis for Single-user RIS Systems in Ricean and Rayleigh Environments}

\author{\IEEEauthorblockN{%
		Ikram Singh\IEEEauthorrefmark{1}, %
		Peter J. Smith\IEEEauthorrefmark{2}, %
		Pawel A. Dmochowski\IEEEauthorrefmark{1}}
	\IEEEauthorblockA{\IEEEauthorrefmark{1}%
		School of Engineering and Computer Science, Victoria University of Wellington, Wellington, New Zealand}
	\IEEEauthorblockA{\IEEEauthorrefmark{2}%
		School of Mathematics and Statistics, Victoria University of Wellington, Wellington, New Zealand}
	\IEEEauthorblockA{email:%
		~\{ikram.singh,peter.smith,pawel.dmochowski\}@ecs.vuw.ac.nz
	}%
}

\maketitle
\begin{abstract}
We present an analysis of the optimal uplink (UL) SNR of a SIMO Reconfigurable Intelligent Surface (RIS)-aided wireless link. We assume that the channel between base station (BS) and RIS is a rank-1 LOS channel while the user (UE)-RIS and UE-BS channels are correlated Ricean. For the optimal RIS matrix, we derive an exact closed form expression for the mean SNR and an approximation for the SNR variance leading to an accurate gamma approximation to the distribution of the UL SNR. Furthermore, we analytically characterise the effects of correlation and the Ricean K-factor on SNR, showing that increasing the K-factor and correlation in the UE-BS channel can have negative effects on the mean SNR, while increasing the K-factor and correlation in the UE-RIS channel improves system performance. We also present favourable and unfavourable channel scenarios which provide insight into the sort of environments that improve or degrade the mean SNR. We also show that the relative gain in the mean SNR when transitioning from an unfavourable to a favourable environment saturates to $(4-\pi)/\pi$ as $N \rightarrow \infty $.
\end{abstract}
\IEEEpeerreviewmaketitle
%
%
\section{Introduction}
Reconfigurable Intelligent Surface (RIS) aided wireless networks are currently the subject of considerable research attention due to their ability to manipulate the channel between users (UEs) and base station (BS) via the RIS. Assuming that channel state information (CSI) is known at the RIS, one can intelligently alter the RIS phases, essentially changing the channel to improve system performance. Here, we focus on a single user system and assume a common system scenario where a RIS is carefully located near the BS such that a rank-1 line-of-sight (LOS) channel is formed between the BS and RIS.  
System scenarios with a LOS channel between the BS-RIS and a single-user are also considered in \cite{GMD,Rank_Imp,Max_Min,OptPHI} with motivation for the LOS assumption given in \cite{Max_Min}. All of these existing works aim to enhance the system to achieve some optimal system performance (sum rate, SINR, etc.) by tuning the RIS phases. In particular, \cite{Max_Min} gives a closed form RIS phase solution without the presence of a direct UE-BS channel for a single user setting, while \cite{OptPHI} gives a closed form phase solution with the presence of a direct channel. We note, however, once the optimal RIS has been defined there is no \textit{exact} analysis of the mean SNR and no analysis of correlation impact on the mean SNR in \cite{GMD,Rank_Imp,Max_Min,OptPHI}.

For the UE to RIS and the direct UE to BS links, the presence of scattering is a more reasonable assumption as is spatial correlation in the channels, especially at the RIS where small inter-element spacing may be envisaged. Several papers do consider spatial correlation in the small-scale fading channels \cite{Max_Min,Prop_Gauss,B5G,Nadeem,Trans_Design,StatCSI,Nadeem2,Two_Timescale}, however, these papers are simulation based and no analysis is given on the impact of correlation on the mean SNR. 

Statistical properties of the RIS-aided channel have been investigated in existing literature. For example, \cite{GammaApproxY,Relay_Comp} provide a closed form expression for the mean SNR in the absence of a UE-BS channel with \cite{Relay_Comp} additionally providing a probability density function (PDF) and a cumulative distribution function (CDF) for the distribution of the mean SNR. In \cite{GammaApproxY,Perf_Analy}, an upper bound is given for the ergodic capacity and in \cite{Asymp_Opt} a lower bound is given for the ergodic capacity. However, there are no closed form expressions for the mean SNR and SNR variance for an optimum RIS-aided wireless system, in the presence of correlated Ricean UE-BS, UE-RIS fading channels. We further note that the work presented in this paper is an extension for the work done in \cite{ISinghRayleigh} which considered correlated Rayleigh UE-BS, UE-RIS fading channels.

In this paper, we focus on an analysis of the optimal uplink (UL) SNR for a single user RIS aided link with a rank-1 LOS RIS-BS channel and correlated Ricean fading for the UE-BS and UE-RIS channels. The contributions of this paper are as follows:
\begin{itemize}
	\item An exact closed-form result for mean SNR and an approximate closed form expression for SNR variance are derived. These are used to show that a gamma distribution provides a good approximation of the UL optimal SNR distribution. Furthermore, using the analysis presented, we are able to reduce the mean SNR and SNR variance expressions when the UE-BS and UE-RIS links are correlated Rayleigh to agree with those in \cite{ISinghRayleigh}.
	\item The analysis is leveraged to gain insight into the impact of the K-factor and spatial correlation on the mean SNR. We show that increasing the K-factor and correlation in the UE-BS channel has negative effects, while increasing the K-factor and correlation in the UE-RIS channel improves the mean SNR. 
	\item Given the analyses, we present favourable and unfavourable channel scenarios, which provide insight into the sort of environments that would improve and degrade the mean SNR. For systems with a large number of RIS elements, we show that when changing from favourable to unfavourable channel scenarios, improvements in the mean SNR saturate at a relative gain of $(4-\pi)/\pi$ as $N\rightarrow\infty$. 
\end{itemize}
\textit{Notation:} $\Exp{\cdot}$ represents statistical expectation. $\Re\left\{ \cdot \right\}$ is the Real operator. $\Norm{\cdot}{2}{}$ denotes the $\ell_{2}$ norm. Upper and lower boldface letters represent matrices and vectors, respectively. $\mathcal{CN}(\boldsymbol{\mu},\mathbf{Q})$ denotes a complex Gaussian distribution with mean $\boldsymbol{\mu}$ and covariance matrix $\mathbf{Q}$. $\mathcal{U}[a,b]$ denotes a uniform random variable taking on values between $a$ and $b$. $\chi_{k}^{2}$ denotes a chi-squared distribution with $k$ degrees of freedom. $\mathbf{1}_n$ represents an $n \times n$ matrix with unit entries. The transpose, Hermitian transpose and complex conjugate operators are denoted as $(\cdot)^{T},(\cdot)^{H},(\cdot)^{*}$, respectively. The trace and diagonal operators are denoted by $\tr{\cdot}{}{}$ and $\text{diag}\{\cdot\}$, respectively. The angle of a vector $\myVM{x}{}{}$ of length $N$ is defined as $\angle \myVM{x}{}{} = [ \angle x_{1},\ldots, \angle x_{N} ]^{T}$ and the exponent of a vector is defined as $e^{\myVM{x}{}{}} = [ e^{x_{1}},\ldots,  e^{x_{N}} ]^{T}$. $\otimes$ denotes the Kronecker product. $L_{\nu}\{\cdot\}$ denotes a Laguerre function of non-integer degree $\nu$.

\section{System Model}
As shown in Fig.~\ref{Fig: System Model}, we examine a RIS aided single user single input multiple output (SIMO) system where a RIS with $N$ reflective elements is located close to a BS with $M$ antennas such that a rank-1 LOS condition is achieved between the RIS and BS. 
\begin{figure}[h]
	\centering
	\includegraphics[width=7cm]{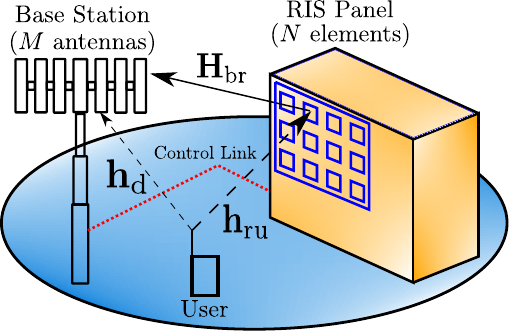}
	\raisecapt\caption{System model (the red dashed line is the control link for the RIS phases).}
	\label{Fig: System Model}
\end{figure}
\subsection{Channel Model}\label{Sec: Channel Model}
Let $\myVM{h}{d}{} \in \mathbb{C}^{M \times 1}$, $\myVM{h}{ru}{} \in \mathbb{C}^{N \times 1}$, $\myVM{H}{br}{} \in \mathbb{C}^{M \times N}$ be the UE-BS, UE-RIS and RIS-BS channels, respectively. The diagonal matrix $\myVM{\Phi}{}{} \in \mathbb{C}^{N \times N}$, where $\mathbf{\Phi}_{rr} = e^{j\phi_{r}}$ for $r=1,2,\ldots,N$, contains the reflection coefficients for each RIS element. The global UL channel is thus represented by
\begin{equation}\label{Eq: 1}
	\mathbf{h} = \myVM{h}{d}{} + \myVM{H}{br}{} \myVM{\Phi}{}{} \myVM{h}{ru}{},
\end{equation}
where we consider the correlated Ricean channels:  
\begin{align*}
	\myVM{h}{d}{} &= \sqrt{\betad}\left( \etad\myVM{a}{d}{} + \zetad\myVM{R}{d}{1/2}\myVM{u}{d}{} \right) \triangleq \sqrt{\betad}\tilde{\mathbf{h}}_{\mathrm{d}}\\
	\myVM{h}{ru}{} &= \sqrt{\betaru}\left( \etaru\myVM{a}{ru}{} + \zetaru\myVM{R}{ru}{1/2}\myVM{u}{ru}{} \right) \triangleq \sqrt{\betaru}\tilde{\mathbf{h}}_{\mathrm{ru}}
\end{align*}
with
\begin{align*}
	\etad = \sqrt{\frac{\kd}{1 + \kd}}, \qquad \zetad = \sqrt{\frac{1}{1 + \kd}}, \\
	\etaru = \sqrt{\frac{\kru}{1 + \kru}}, \qquad \zetaru = \sqrt{\frac{1}{1 + \kru}},
\end{align*}
where $\beta_{\textrm{d}}$ and $\beta_{\textrm{ru}}$ are the link gains between UE-BS and UE-RIS respectively, $\myVM{R}{d}{}$ and $\myVM{R}{ru}{}$ are the correlation matrices for UE-BS and UE-RIS links, respectively, and $\myVM{u}{d}{},\myVM{u}{ru}{} \sim \mathcal{CN}(\mathbf{0},\mathbf{I})$. Here, $\kd$ and $\kru$ are the Ricean K-factors for the UE-BS and UE-RIS channels, respectively. The LOS paths $\myVM{a}{d}{}$ and $\myVM{a}{ru}{}$ are topology specific steering vectors at the BS and RIS respectively. Particular examples of steering vectors for a vertical uniform rectangular array (VURA) are given in Sec.~\ref{Sec: Results}.

Note, that the correlation matrices $\myVM{R}{ru}{}$ and $\myVM{R}{d}{}$ can represent any correlation model. For simulation purposes, we will use the well-known exponential decay model,
\begin{equation}
(\myVM{R}{ru}{})_{ik}= \rho_{\mathrm{ru}}^{\frac{d_{i,k}}{d_{\textrm{r}}}} 
\quad , \quad
(\myVM{R}{d}{})_{ik} = \rho_{\mathrm{d}}^{\frac{d_{i,k}}{d_{\textrm{b}}}},
\end{equation}
where $0 \leq \rho_{\text{ru}} \leq 1$, $0 \leq \rho_{\text{d}} \leq 1$. $d_{i,k}$ is the distance between the $i^{\textrm{th}}$ and $k^{\textrm{th}}$ antenna/element at the BS/RIS. $d_{\textrm{b}}$ and $d_{\textrm{r}}$ are the nearest-neighbour BS antenna separation and RIS element separation, respectively, which are measured in wavelength units. $\rho_{\mathrm{d}}$ and $\rho_{\mathrm{ru}}$ are the nearest neighbour BS antenna and RIS element correlations, respectively.

The rank-1 LOS channel from RIS to BS has link gain $\beta_{\mathrm{br}}$ and is given by $\myVM{H}{br}{} = \sqrt{\beta_{\textrm{br}}}\myVM{a}{b}{}\myVM{a}{r}{H}$ where $\myVM{a}{b}{}$ and $\myVM{a}{r}{}$ are topology specific steering vectors at the BS and RIS, respectively. 

\subsection{Optimal SNR}\label{Sec: Optimal SNR}
Using  (\ref{Eq: 1}), the received signal at the BS is,
$
	\mathbf{r} = \left(\myVM{h}{d}{} + \myVM{H}{br}{} \myVM{\Phi}{}{} \myVM{h}{ru}{} \right)s + \mathbf{n} \triangleq \myVM{h}{}{}s + \mathbf{n},
$
where $s$ is the transmitted signal with power $E_{s}$ and $\mathbf{n} \sim \mathcal{CN}(\mathbf{0},\sigma^{2}\mathbf{I})$. For a single user, Matched Filtering (MF) is the optimal combining method so the optimal UL SNR is,
$
	\text{SNR} = \myVM{{h}}{}{H}\myVM{{h}}{}{}\bar{\tau},
$
where $\bar{\tau} = \frac{E_{s}}{\sigma^{2}}$. The optimal RIS phase matrix to maximize the SNR can be computed using the main steps outlined in \cite[Sec. III-B]{OptPHI} but using an UL channel model instead of downlink. Substituting the channel vectors and through some matrix algebraic manipulation, the optimal RIS phase matrix is,

\begin{equation}\label{Eq: Optimum PHI}
	\myVM{\Phi}{}{} = 
	\frac{\myVM{a}{b}{H} \myVM{h}{d}{}}{\Abs{\myVM{a}{b}{H} \myVM{h}{d}{}}{}}
	\Diag{e^{j\angle\myVM{a}{r}{}}}
	\Diag{e^{-j\angle\myVM{h}{ru}{}}}.
\end{equation}
Substituting \eq{\ref{Eq: Optimum PHI}} into $\myVM{h}{}{}$, the optimal UL global channel is,
\begin{equation}\label{Eq: Final r signal}
	\myVM{h}{}{} =  \myVM{h}{d}{} + \sqrt{\betabr\betaru} \psi \sum_{n=1}^{N}\Abs{\tilde{\mathbf{h}}_{\mathrm{ru},n}}{} \myVM{a}{b}{} ,
\end{equation}
where $\psi = \frac{\myVM{a}{b}{H} \myVM{h}{d}{}}{\Abs{\myVM{a}{b}{H} \myVM{h}{d}{}}{}}$, giving the optimal UL SNR as
\begin{equation}\label{Eq: SNR Eq}
	\text{SNR} = 
	\left( \myVM{h}{d}{H} + \alpha^{*}\myVM{a}{b}{H} \right)
	\left( \myVM{h}{d}{} + \alpha^{}\myVM{a}{b}{} \right)
	\bar{\tau},
\end{equation}
where $\alpha = \sqrt{\beta_{\textrm{br}}\beta_{\textrm{ru}}} \psi Y$ and $Y=\sum_{n=1}^{N}\Abs{\tilde{\mathbf{h}}_{\mathrm{ru},n}}{}$. The $Y$ variable  is given its own notation as it arises in many of the derivations for the SNR.

\section{Mean SNR and SNR Variance}\label{Sec: E{SNR} and Var{SNR}}
\subsection{Correlated Ricean Channels}
Here, we provide an exact closed form result for $\mathbb{E}\{\text{SNR}\}$ and an exact expression for $\Var{\text{SNR}}$ when $\myVM{h}{d}{}$ and $\myVM{h}{ru}{}$ are correlated Ricean channels. Utilising these expressions, we then show that they can be reduced to the expressions given in \cite{ISinghRayleigh} when $\myVM{h}{d}{}$ and $\myVM{h}{ru}{}$ are correlated Rayleigh channels.
\begin{theorem}
	The mean SNR is given by
	\begin{align}\label{Eq: Corr E{SNR} Ricean}
		& \Exp{\mathrm{SNR}} = \notag \\
		& \Bigg( \frac{NA\zetad\zetaru\pi\sqrt{\betad\betabr\betaru}}{2}
		\LagP{-\kru}\LagP{-\frac{\kd \Abs{\myVM{a}{b}{H}\ad{}}{2}}{A^{2}}} \notag \\
		& \quad + \beta_{\mathrm{d}}M + \beta_{\mathrm{br}}\beta_{\mathrm{ru}}M(N+F_{\mathrm{R}}) \Bigg) \bar{\tau},
	\end{align}
	with $F_{R}$ given by
	\begin{align}\label{Eq: Whittaker M}
		& F_{R} = 
		\underset{i \neq k}{\sum_{i=1}^{N} \sum_{k=1}^{N}} 
		\frac{\left(1 - \Abs{\rho_{ik}}{2}\right)^2}{1 + \kru} 
		\mathrm{exp}\left\{ -\frac{2\kru - 2\mu_c\kru}{1 - \Abs{\rho_{ik}}{2}}\right\} \notag \\
		& \times \sum_{m=0}^{\infty}\sum_{n=0}^{m} \frac{\epsilon_n \cos(n\phi) \Abs{\rho_{ik}}{2m - n}}{m! (m-n)! (n!)^2} \left( \frac{\kru(1 + \Abs{\rho_{ik}}{2} - 2\mu_c)}{1 - \Abs{\rho_{ik}}{2}} \right)^n \notag \\
		& \times \Gamma^2\left( m + \frac{3}{2}\right) 
		{}_{1}F_{1}^2\left(m+\frac{3}{2}, n+1, \frac{\kru(1 + \Abs{\rho_{ik}}{2} - 2\mu_c)}{1 - \Abs{\rho_{ik}}{2}}  \right)
	\end{align}
	and $A = \left\lVert \myVM{R}{d}{1/2} \myVM{a}{b}{} \right\rVert _{2}$, where $\rho_{ik} = \left(\myVM{R}{ru}{}\right)_{ik}$,  ${}_{1}F_{1}(\cdot)$ is the confluent hypergeometric function, $\phi = \arg\{ (1 + \Abs{\rho_{ik}}{2})\mu_c\kru - 2\kru\Abs{\rho_{ik}}{2} + j(1 - \Abs{\rho_{ik}}{2}\mu_s\kru)\}$, $\mu_c = \Abs{\rho_{ik}}{}\cos(\Delta\theta)$, $\mu_s = \Abs{\rho_{ik}}{}\sin(\Delta\theta)$, $\Delta\theta = \arg\{\mathbf{\aru{}}_{,i}\} - \arg\{\mathbf{\aru{}}_{,k}\}$ and $\epsilon_n$ is given by $\epsilon_0=1$ and $\epsilon_n=2$ for $n \ge 1$ \cite[Eq.~(26b)]{Yacoub}.
\end{theorem}
\begin{proof}
	See App.~\ref{AppA: E{SNR} Corr Deri} for the derivation of \eq{\ref{Eq: Corr E{SNR} Ricean}}.
\end{proof}
The formula for the mean SNR in \eqref{Eq: Corr E{SNR} Ricean} is compact and intuitive except for the term $F_R$ which is easier to understand from its definition in App.~\ref{AppA: E{SNR} Corr Deri} as
\begin{equation}\label{FR_defn}
F_R = \underset{i \neq j}{\sum_{i=1}^{N} \sum_{k=1}^{N}} \mathbb{E}\left\{ \Abs{\myVMIndex{\tilde{h}}{ru}{i}{}}{} \Abs{\myVMIndex{\tilde{h}}{ru}{k}{}}{}  \right\}. 
\end{equation} 
Note that evaluation of $F_R$ requires the product moment of two correlated Ricean variables. While compact expressions are given in \cite{Zogas} for the PDF of two correlated Ricean variables, it is important to note that these PDFs are for the special case where the LOS components are the same. Here,  there is  a phase shift that occurs between any two RIS elements and this assumption cannot be made. Hence, we have used the more general bivariate Ricean distribution in \cite{Yacoub}. When the UE-RIS channel is non-LOS (correlated Rayleigh) then the complexity in all the expressions reduces considerably.
\begin{theorem}
	The SNR Variance is given by
	\begin{align}\label{Eq: Corr Var{SNR} Ricean}
		& \Var{\mathrm{SNR}} = \notag \\
		& \Bigg(\betad^{2}\zetad^{2}\left[ 2\etad^{2}\Norm{\myVM{R}{d}{1/2}\myVM{a}{d}{}}{2}{2} + \zetad^{2}\tr{\myVM{R}{d}{2}}{}{}\right] \notag \\
		& + \betad^{3/2}\sqrt{\betabr\betaru} A N L_{\mathrm{ru}} \pi \left[\frac{2}{\sqrt{\pi}}(T_{21}+T_{22}+T_{23}) - M L_{\mathrm{d}}\right] \notag \\
		& + \betad\betabr\betaru\Bigg[ 4\left( \etad^{2}\Abs{\myVM{a}{b}{H}\ad{}}{2} + \zetad^{2}A^{2} \right)(N+F_{\mathrm{R}}) \notag \\
		& \qquad -\frac{N^2 A^2 \pi^2 L_{\mathrm{ru}}^{2} L_{\mathrm{d}}^{2} }{4}\Bigg] \notag \\
		& + \sqrt{\betad}(\betabr\betaru)^{3/2} M A L_{\mathrm{d}}\left[ 2\sqrt{\pi}C_{1} - N(N+F_{\mathrm{R}})\pi L_{\mathrm{ru}} \right] \notag \\
		& + (M\betabr\betaru)^{2}\left[ C_{2} - (N + F_{R})^{2} \right] \Bigg) \bar{\tau}^{2}
	\end{align}
with
\begin{align}
	& L_{\mathrm{ru}} = \zetaru \LagP{-\kru} \label{Eq: Lru} \\
	& L_{\mathrm{d}} = \zetad \LagP{-\kd \frac{\Abs{\myVM{a}{b}{H}\ad{}}{2}}{A^2}} \\
	& T_{21} = \frac{3B\zetad^{3}\sqrt{\pi}}{4} L_{3/2}\left(-C\right)  
	- \frac{2B\etad}{A}\Re\{\myVM{a}{b}{H}\ad{} \mathcal{I} \} \notag \\
	& \quad + \frac{L_{\mathrm{d}}}{2} \left[ \frac{B\kd\Abs{\ad{H}\myVM{a}{b}{}}{2}}{A^2 (1+\kd)} + \frac{\sqrt{\pi}(M-B)}{1+\kd}\right] 
\end{align}
\begin{align}
	& T_{22} = \frac{2\etad}{A}
	\Re\left\{
	\ad{H}\myVM{R}{d}{}\myVM{a}{b}{} \left[ \mathcal{I} - \frac{\etad\ad{H}\myVM{a}{b}{}\sqrt{\pi}}{2A} L_{\mathrm{d}} \right] \right\} \\
	& T_{23} = \frac{\etad^{2} M \sqrt{\pi}}{2} L_{\mathrm{d}} \\
	& \mathcal{I} = 
	\frac{-3 A \sqrt{\pi}}{4\sqrt{\kd}(1+\kd)} 
	\exp\left\{ j\angle \ad{H}\myVM{a}{b}{} - \frac{\kd\Abs{\ad{H}\myVM{a}{b}{}}{2}}{2 A^2} \right\} \notag \\
	& \quad \times M_{3/2,1/2}\left( -\frac{\kd\Abs{\ad{H}\myVM{a}{b}{}}{2}}{A^2} \right) \label{Eq: Fancy I}
\end{align}
where $A$ is given in Theorem 1, $F_{\mathrm{R}}$ is given by \eq{\ref{Eq: Whittaker M}}, $B = \Norm{\myVM{R}{d}{}\myVM{a}{b}{}}{2}{2}/A^2$, $C = \kd\Abs{\ad{H}\myVM{a}{b}{}}{2}/A^{2}$, $M_{3/2,1/2}(\cdot)$ is the Whittaker M function, $C_{1}=\Exp{Y^{3}}$ and $C_{2}=\Exp{Y^{4}}$.
\end{theorem}
\begin{proof}
	See App.~\ref{App: Var{SNR} Corr Deri} and App.~\ref{App: Curly I} for the derivation of \eq{\ref{Eq: Corr Var{SNR} Ricean}}.
\end{proof}
In \eq{\ref{Eq: Corr Var{SNR} Ricean}}, all terms are known except for $C_{1}=\Exp{Y^3}$ and $C_{2}=\Exp{Y^4}$. To the best of our knowledge, these moments are intractable and approximations are required for these moments. Note that $Y$ is positive, unimodal, and the sum of $N$ random variables. A gamma distribution is commonly used to approximate such distributions, eg., \cite{GammaApproxY}. Since the first and second moments of $Y$ are known exactly, we are able to fit a gamma approximation to $Y$ with the correct first and second moments. Using this gamma approximation, we obtain approximations for $\Exp{Y^3}$ and $\Exp{Y^4}$ as presented in the following Corollary,
\begin{corollary}\label{Corollary}
	To obtain an approximation of the SNR variance in closed form, we approximate the 3$^{\text{rd}}$ and 4$^{\text{th}}$ moments of $Y$ by,
	\begin{align}\label{Eq: Corollary Ricean}
		\hspace{ -5pt}\Exp{Y^3}  & =  b^3 a \prod_{k=1}^{2}(k+a), \quad
		\Exp{Y^4}  =  b^4 a \prod_{k=1}^{3}(k+a)
	\end{align}
	with 
	\begin{align*}
		a & = \dfrac{N^{2} \pi L_{\mathrm{ru}}^{2}}{4(N+F_{R}) - N^{2} \pi L_{\mathrm{ru}}^{2}} \\
		b & = \dfrac{2}{N \sqrt{\pi} L_{\mathrm{ru}}}\left( N + F_{R} - \frac{N^2 \pi^2}{4} L_{\mathrm{ru}}^{2} \right)
	\end{align*}
	where $F_{\mathrm{R}}$ is given by \eq{\ref{Eq: Whittaker M}} and $L_{\mathrm{ru}}$ is given by \eq{\ref{Eq: Lru}}.
\end{corollary}
\begin{proof}See App.~\ref{App: E{Y^3},E{Y^4} Approx} for the derivation of \eq{\ref{Eq: Corollary Ricean}}. \end{proof}

Note that for high values of correlation and the K-factor, \eq{\ref{Eq: Whittaker M}} is computationally expensive as in such scenarios the number of terms required in the double summation over $m$ and $n$ grows large. In these situations, an alternative approach is to replace the double summation with its integral equivalent, derived in App.~\ref{App: Integral form of E{ri rj}}, and presented in the following cases where $\rhoruij{}<1$  and  $\rhoruij{} = 1$. The case of perfect correlation,  $\rhoruij{} = 1$, provides a useful benchmark to evaluate the SNR trends as the level of correlation increases.

\subsection{Case 1: $F_R$ for $\rhoruij{}<1$}\label{SubSec: FR for high kru and rhoru}
\begin{align}\label{Eq: FR for high kru and rhoru}
F_{R} &= 
\underset{i \neq k}{\sum_{i=1}^{N} \sum_{k=1}^{N}} \frac{\zetaru^2 \sqrt{1 - \rhoruij{2}}}{2\sqrt{\pi}} \int_{0}^{2\pi}\int_{0}^{\infty} r e^{-r^2} \notag \\
& \times \sqrt{\Abs{a_i}{2} + r^2 + 2r\Abs{a_i}{}\cos(\theta - \theta_i)} \notag \\
& \times \LagP{-b(\Abs{a_k}{2} + r^2 + 2r\Abs{a_k}{}\cos(\theta - \theta_k))} \ d\theta dr,
\end{align}
where $a_i = \etaru \myVMIndex{a}{ru}{i}{}/\zetaru$, $a_k = \etaru\myVMIndex{a}{ru}{k}{}/(\zetaru\rhoruij{})$, $\theta_i = \angle \myVMIndex{a}{ru}{i}{}$ and $\theta_k = \angle \myVMIndex{a}{ru}{k}{}$.

\subsection{Benchmark Case: $F_R$ when $\rhoruij{} = 1$}\label{SubSec: FR for rhoru = 1}
For the benchmark case of maximum correlation, $\rhoruij{} = 1$, both \eq{\ref{Eq: Whittaker M}} and \eq{\ref{Eq: FR for high kru and rhoru}} result in an indeterminate answer. However, we can use the integral form of $F_R$ and find its result as $\rhoruij{} \rightarrow 1$ which, using the derivation in App.~\ref{App: Integral form of E{ri rj}}, gives 
\begin{align}\label{Eq: FR for rhoru = 1}
F_{R} &= 
\frac{\zetaru^2}{\pi} \underset{i \neq k}{\sum_{i=1}^{N} \sum_{k=1}^{N}} \int_{0}^{2\pi}\int_{0}^{\infty} r e^{-r^2} \notag \\
& \hspace{1em} \times \sqrt{\Abs{a_i}{2} + r^2 + 2r\Abs{a_i}{}\cos(\theta - \theta_i)} \notag \\
& \hspace{1em} \times \sqrt{\Abs{a_i}{2} + r^2 + 2r\Abs{a_i}{}\cos(\theta - \theta_k)} \ d\theta dr,
\end{align}
where $a_i, \theta_i$ and $\theta_k$ are defined in Sec.~\ref{SubSec: FR for high kru and rhoru}. Note that the double numeric integrations required in \eq{\ref{Eq: FR for high kru and rhoru}} and \eq{\ref{Eq: FR for rhoru = 1}} are computationally convenient as: one integral is finite range and the integrand is smooth and rapidly decaying as $r \rightarrow \infty$.

\subsection{Special Case 1: Uncorrelated Ricean Case}\label{SubSec: Special Case: UnCorr Ricean}
For independent Ricean fading, we provide exact closed form expressions for both $\Exp{\mathrm{SNR}}$ and $\Var{\mathrm{SNR}}$. 
The mean SNR is given by
\begin{align}\label{Eq: Uncorr E{SNR} Ricean}
	& \Exp{\mathrm{SNR}} = \notag \\
	& \Bigg( \frac{N\sqrt{M}\zetad\zetaru\pi\sqrt{\betad\betabr\betaru}}{2}
	\LagP{-\kru}\LagP{-\frac{\kd \Abs{\myVM{a}{b}{H}\ad{}}{2}}{M}} \notag \\
	& \quad + \beta_{\mathrm{d}}M + \beta_{\mathrm{br}}\beta_{\mathrm{ru}}M(N+F_{\mathrm{R}}) \Bigg) \bar{\tau},
\end{align}
with $F_{R}$ given by
\begin{align}\label{Eq: FR uncorr Ricean}
	F_{R} = \frac{\pi N(N-1)}{4} \left(\zetaru \LagP{-\kru}\right)^2.
\end{align}
This can be easily derived from Theorem 1 by noting that $F_R$ is the sum of products of iid Ricean random variables over all $1 \leq i,k \leq N$ for $i \neq k$ and $A = \sqrt{M}$ since $\myVM{R}{d}{} = \myVM{I}{M}{}$ and $\myVM{R}{ru}{} = \myVM{I}{N}{}$.

The SNR Variance is given by
\begingroup
\allowdisplaybreaks
	\begin{align}\label{Eq: Corr Var{SNR} Ricean_uncorr}
		& \Var{\mathrm{SNR}} = \notag \\
		& \Bigg(\betad^{2}\zetad^{2}\left[ 2\etad^{2}M + \zetad^{2}M\right] \notag \\
		& + \betad^{3/2}\sqrt{\betabr\betaru} \sqrt{M} N L_{\mathrm{ru}} \pi \left[\frac{2}{\sqrt{\pi}}(T_{21}+T_{22}+T_{23}) - M L_{\mathrm{d}}\right] \notag \\
		& + \betad\betabr\betaru\Bigg[ 4\left( \etad^{2}\Abs{\myVM{a}{b}{H}\ad{}}{2} + \zetad^{2}M \right)(N+F_{\mathrm{R}}) \notag \\
		& \qquad -\frac{N^2 M \pi^2 L_{\mathrm{ru}}^{2} L_{\mathrm{d}}^{2} }{4}\Bigg] \notag \\
		& + \sqrt{\betad}(\betabr\betaru)^{3/2} M^{3/2} L_{\mathrm{d}}\left[ 2\sqrt{\pi}C_{1\mathrm{u}} - N(N+F_{\mathrm{R}})\pi L_{\mathrm{ru}} \right] \notag \\
		& + (M\betabr\betaru)^{2}\left[ C_{2\mathrm{u}} - (N + F_{R})^{2} \right] \Bigg) \bar{\tau}^{2},
	\end{align}
with
\begin{align}
	& L_{\mathrm{ru}} = \zetaru \LagP{-\kru} \label{Eq: Lru2} \\
	& L_{\mathrm{d}} = \zetad \LagP{-\kd \frac{\Abs{\myVM{a}{b}{H}\ad{}}{2}}{M}} \\
	& T_{21} = \frac{3B\zetad^{3}\sqrt{\pi}}{4} L_{3/2}\left(-C\right)  
	- \frac{2B\etad}{\sqrt{M}}\Re\{\myVM{a}{b}{H}\ad{} \mathcal{I} \} \notag \\
	& \quad + \frac{L_{\mathrm{d}}}{2} \left[ \frac{B\kd\Abs{\ad{H}\myVM{a}{b}{}}{2}}{M (1+\kd)} + \frac{\sqrt{\pi}(M-B)}{1+\kd}\right] 
\end{align}
\begin{align}
	& T_{22} = \frac{2\etad}{\sqrt{M}}
	\Re\left\{
	\ad{H}\myVM{R}{d}{}\myVM{a}{b}{} \left[ \mathcal{I} - \frac{\etad\ad{H}\myVM{a}{b}{}\sqrt{\pi}}{2\sqrt{M}} L_{\mathrm{d}} \right] \right\} \\
	& T_{23} = \frac{\etad^{2} M \sqrt{\pi}}{2} L_{\mathrm{d}} \\
	& \mathcal{I} = 
	\frac{-3 \sqrt{M} \sqrt{\pi}}{4\sqrt{\kd}(1+\kd)} 
	\exp\left\{ j\angle \ad{H}\myVM{a}{b}{} - \frac{\kd\Abs{\ad{H}\myVM{a}{b}{}}{2}}{2 M} \right\} \notag \\
	& \quad \times M_{3/2,1/2}\left( -\frac{\kd\Abs{\ad{H}\myVM{a}{b}{}}{2}}{M} \right) \label{Eq: Fancy I2} \\
	& C_{\mathrm{1u}} = \frac{3 \sqrt{\pi} N}{4}L_{\mathrm{ru1}} + \frac{\pi^{3/2} N!}{8(N-3)!} L_{\mathrm{ru}}^3 + \frac{3 \sqrt{\pi} N!}{2(N-2)!} L_{\mathrm{ru}} \label{Eq: C1u}\\
	&C_{\mathrm{2u}} = N(2\zetaru^2 + \etaru^2(4\zetaru + \etaru^2)) + \frac{\pi^2 N!}{16(N-4)!} L_{\mathrm{ru}}^4 \notag \\
	& \quad + \frac{3 N!}{(N-2)!}\left(1 + \sqrt{\pi}L_{\mathrm{ru1}}\right) + \frac{3\pi N!}{2(N-3)!}L_{\mathrm{ru}}^2. \label{Eq: C2u}
\end{align}
\endgroup
Equations \eqref{Eq: Corr Var{SNR} Ricean_uncorr}-\eqref{Eq: C2u} follow from Theorem 2 using the substitutions $A=\sqrt{M}$, $\Norm{\myVM{a}{d}{}}{2}{2} = M$, $\tr{\myVM{R}{d}{2}}{}{} = M$ and the version of $F_R$ given in \eqref{Eq: FR uncorr Ricean}.
A consequence of having uncorrelated UE-BS and UE-RIS channels, is that the expectations $\Exp{Y^{3}}$ and $\Exp{Y^{4}}$ are known exactly and these are given by \eq{\ref{Eq: C1u}} and \eq{\ref{Eq: C2u}} respectively.

\subsection{Special Case 2: Correlated Rayleigh Case}\label{SubSec: Special Case: Corr Rayleigh}
We obtain mean SNR and SNR variance expressions for correlated Rayleigh channels $\myVM{h}{d}{}$ and $\myVM{h}{ru}{}$ by setting $\kd = \kru = 0$.
\subsubsection{Mean SNR}\label{SubSubSec: Mean SNR Corr Rayleigh}
Note that when $\kru=0$, $\zetaru=1$ and $\LagP{-\kru}=1$. Similarly when $\kd=0$, $\zetad=1$ and $\LagP{-\frac{\kd \Abs{\myVM{a}{b}{H}\ad{}}{2}}{A^{2}}}=1$. As such, the second term in \eq{\ref{Eq: Corr E{SNR} Ricean}} reduces down to the second term in \cite[Eq.~(8)]{ISinghRayleigh}. The third term can be reduced to its correlated Rayleigh form as shown in \cite[Eq.~(8)]{ISinghRayleigh} by reducing $F_{R}$ to \eq{\ref{Eq: HyperGeometric}} (see App.~\ref{App: Reduce FR to Corr Rayleigh}). The final form after reducing all of the terms in \eq{\ref{Eq: Corr E{SNR} Ricean}} results in the mean SNR expression in \cite[Eq.~(8)]{ISinghRayleigh}; given below for completeness:
\begin{theorem}
	The mean SNR when $\myVM{h}{d}{}$ and $\myVM{h}{ru}{}$ are correlated Rayleigh channels is given by
	\begin{align}\label{Eq: Corr E{SNR}}
	& \Exp{\mathrm{SNR}} = \notag \\
	& \left( 
	\hspace{-0.3em}\beta_{\mathrm{d}}M  +
	\dfrac{NA\pi\sqrt{\beta_{\mathrm{d}}\beta_{\mathrm{br}}\beta_{\mathrm{ru}}} }{2}  
	+ \beta_{\mathrm{br}}\beta_{\mathrm{ru}}M(N + F) 
	\right)\bar{\tau},
	\end{align}
	with $F$ given by
	\begin{equation}\label{Eq: HyperGeometric}
	F = \underset{i \neq k}{\sum_{i=1}^{N} \sum_{k=1}^{N}} \dfrac{\pi}{4}\left( 1 - \left\lvert\rho_{ik}\right\rvert^2\right)^2 {}_{2}F_{1}\left(\frac{3}{2},\frac{3}{2};1;\left\lvert\rho_{ik}\right\rvert^2 \right),
	\end{equation}
	and $A = \left\lVert \myVM{R}{d}{1/2} \myVM{a}{b}{} \right\rVert _{2}$, where ${}_{2}F_{1}(\cdot)$ is the Gaussian hypergeometric function and $\rho_{ij} = \left(\myVM{R}{ru}{} \right)_{ij}$. 
\end{theorem}

\subsubsection{SNR Variance}\label{SubSubSec: Var SNR Corr Rayleigh}
In addition to the steps undertaken in \ref{SubSubSec: Mean SNR Corr Rayleigh}, we further note that $\etad = \etaru = 0$. This eliminates many of the terms in \eq{\ref{Eq: Corr Var{SNR} Ricean}} of Theorem 2. Using App.~\ref{App: Reduce FR to Corr Rayleigh} to collapse $F_{R}$ down to its correlated Rayleigh form, we obtain the final SNR variance expression for correlated Rayleigh channels, consistent with that given in \cite[Eq.~(10)]{ISinghRayleigh} and given below for completeness:
\begin{theorem}
	The SNR variance when $\myVM{h}{d}{}$ and $\myVM{h}{ru}{}$ are correlated Rayleigh channels is given by
	\begin{align}\label{Eq: Corr var{SNR}}
	\Var{\mathrm{SNR}}&=\Big(\beta_{\mathrm{d}}^{2} \tr{R}{d}{2} 
	+ \beta_{\mathrm{d}}^{3/2}\sqrt{\beta_{\mathrm{br}}\beta_{\mathrm{ru}}}N\pi(B - MA) \notag \\
	& + \beta_{\mathrm{d}}\beta_{\mathrm{br}}\beta_{\mathrm{ru}} A^{2} \left(4(N+F) - \dfrac{N^2\pi^2}{4} \right)\notag \\
	& + M A \sqrt{\beta_{\mathrm{d}}}\left(\beta_{\mathrm{br}}\beta_{\mathrm{ru}}\right)^{3/2} \left( 2\sqrt{\pi}C_{1} - N(N+F)\pi\right) 	\notag \\
	& + (M\beta_{\mathrm{br}}\beta_{\mathrm{ru}})^2 \left(C_{2} - (N+F)^2\right) \Big) \bar{\tau}^{2} ,
	\end{align}
	where $ B = MA +{\myVM{a}{b}{H} \myVM{R}{d}{2} \myVM{a}{b}{}}/{2A} $, $F$ is given by (\ref{Eq: HyperGeometric}), $A$ is given in Theorem 1 and $C_{1}=\Exp{Y^{3}},C_{2}=\Exp{Y^{4}}$. 
\end{theorem}
\begin{corollary}
	To obtain an SNR variance approximation in closed form, we approximate the 3$^{\text{rd}}$, 4$^{\text{th}}$ moments of $Y$ by,
	\begin{align}\label{Eq: Corollary Rayleigh}
	\hspace{ -5pt}\Exp{Y^3}  & =  b^3 a \prod_{k=1}^{2}(k+a), \quad
	\Exp{Y^4}  =  b^4 a \prod_{k=1}^{3}(k+a)
	\end{align}
	with 
	\begin{equation*}
	a = \dfrac{N^{2}\pi}{4(N+F) - N^{2}\pi} \quad,\quad
	b = \dfrac{2}{N\sqrt{\pi}}\left( N + F - \dfrac{N^{2}\pi}{4} \right).
	\end{equation*}
	where $F$ is given by \eq{\ref{Eq: HyperGeometric}}.
\end{corollary}
\section{Performance Insights based on $\Exp{\text{SNR}}$}\label{Sec: Perf. Insights based on E{SNR}}

Note that the first term in the mean SNR in \eq{\ref{Eq: Corr E{SNR} Ricean}} is influenced by variants of the function $q(\alpha,x) \triangleq \LagP{-\alpha x}/\sqrt{1 + x}$ for $x \in \{\kru,\kd\}$, $\alpha \geq 0$,  where the factor of $1/\sqrt{1 + x}$ arises  from the $\zetaru$ and $\zetad$ terms. Hence, we study the behaviour of $q(\alpha,x)$ before developing performance insights based on the mean SNR.
\subsection{Behaviour of $\LagP{-\alpha x}/\sqrt{1 + x}$}\label{SubSec: Behaviour of Laguerre Function}
Using \cite[Eq. (13.6.9)]{Stegun} the Laguerre function can be rewritten in terms of the confluent hypergeometric function, ${}_{1}F_{1}$, as
\begin{equation}\label{Eq: Confluent Hyper. representation of Laguerre fuction}
	\LagP{-\alpha x} =
	{}_{1}F_{1}\left(-\frac{1}{2};1;\alpha x\right).
\end{equation}
Using the series expansion for ${}_{1}F_{1}$ \cite[Eq. (13.1.2)]{Stegun} and the Maclaurin series expansion of $1/\sqrt{1 + x}$, we have the following expansion for $q(\alpha,x)$,
\begin{align*}
	q(\alpha,x) 
	&=
	\left( 1 + \frac{\alpha x}{2} - 
	\frac{\alpha^2 x^2}{16} + \ldots
	\right)
	\left( 1 - \frac{x}{2} + 
	\frac{3x^2}{8} - \ldots
	\right) \\
	&=
	1 + \frac{\alpha-1}{2}x + 
	\frac{6-4\alpha-\alpha^2}{16}x^2 + \ldots
\end{align*}
First, we look at the behaviour near the origin. Noting that $q(\alpha,0) = 1$, then for $\alpha \geq 1$, $q(\alpha,x)$ is increasing in $x$  and for $\alpha<1$, $q(\alpha,x)$ is decreasing in $x$.

Next, we look at the behaviour as $x \to \infty$. Using the asymptotic behaviour of ${}_{1}F_{1}$ \cite[Eq. (13.1.5)]{Stegun}, we have 
\begin{align}
	q(\alpha,x) &\sim
	2\sqrt{\frac{\alpha x}{\pi(1+x)}}\left(1 + \frac{1}{4\alpha x}\right) \label{Eq: Asymptotic value} \\
	&\to 2\sqrt{\frac{\alpha}{\pi}}
	\quad \text{ as } x \rightarrow \infty. \label{Eq: Asymptotic value Lim}
\end{align}
From \eq{\ref{Eq: Asymptotic value}}, we see that as $x\rightarrow\infty$, $q(\alpha,x)\geq q(\alpha,0) =1$ for $\alpha \geq \pi/4$ and $q(\alpha,x)<q(\alpha,0) =1$ for $\alpha < \pi/4$. Differentiating \eqref{Eq: Asymptotic value} and  simplifying gives
\begin{equation}\label{deriv}
	\frac{dq(\alpha,x)}{dx} \approx
	\frac{1}{2\alpha x^{3/2}\sqrt{1+x}}\left( \frac{\alpha x + 1}{1+x} - \frac{1}{2} \right),
\end{equation}
for large values of $x$. From \eqref{deriv}, we make the following observations. As $x\rightarrow\infty$, $\frac{dq(\alpha,x)}{dx}<0$ for $\alpha < 1/2$ and $\frac{dq(\alpha,x)}{dx}>0$ for $\alpha \geq  1/2$. 

Collating these properties, we obtain the following result:
\begin{result} The main  characteristics of $q(\alpha,x)$ are given by
$$
q(\alpha,x) \Longrightarrow
\begin{cases}
\text{increases near the origin}, & \text{for } \alpha \geq 1 \\
\text{decreases near the origin}, & \text{for } \alpha < 1 \\
\text{positive gradient as } x\rightarrow\infty, & \text{for } \alpha \geq 1/2 \\
\text{negative gradient as } x\rightarrow\infty, & \text{for } \alpha < 1/2 \\
\geq 1 \text{ as } x\rightarrow\infty, & \text{for } \alpha \geq \pi/4 \\
< 1 \text{ as } x\rightarrow\infty, & \text{for } \alpha < \pi/4 \\
\end{cases}.
$$
\end{result}

\subsection{Effect of $\kru$ on $\Exp{\text{SNR}}$}\label{SubSec: Effect of kru on E{SNR}}


Only the first and third terms in \eq{\ref{Eq: Corr E{SNR} Ricean}} are affected by $\kru$. The first term contains $q(1,\kru)$ so is increasing with $\kru$ from Result 1. In the third term, $\kru$ affects the mean SNR via the $F_R$ expression. Applying H{\"o}lders inequality \cite[E.q. (1.1)]{2244734} to $F_R$ in \eqref{FR_defn} gives $F_R \le N(N-1)$. This upper bound is achieved when $\kru \rightarrow \infty$. Hence, we conclude that increasing $\kru$ improves the mean SNR.




\subsection{Effect of $\kd$ on $\Exp{\text{SNR}}$}\label{SubSec: Effect of kd on E{SNR}}


Only the first term in \eq{\ref{Eq: Corr E{SNR} Ricean}} is affected by $\kd$ through the expression
\begin{equation*}
	\zetad L_{1/2}\left(-\frac{\kd \Abs{\myVM{a}{b}{H}\ad{}}{2}}{A^2}\right) = q\left(\frac{ \Abs{\myVM{a}{b}{H}\ad{}}{2}}{A^2}, \kd\right).
\end{equation*}
From Result 1, this expression increases with $\kd$ when $\Abs{\myVM{a}{b}{H}\ad{}}{}>A$, decreases when $\Abs{\myVM{a}{b}{H}\ad{}}{}<\sqrt{2A}$ and has a more complex behavior in-between. 

This pattern can be explained by the fact that $\Abs{\myVM{a}{b}{H}\ad{}}{}$ measures the alignment of the RIS-BS channel, $\myVM{a}{b}{}$, with the LOS component of the direct path, $\ad{}$, while $A = \left\lVert \myVM{R}{d}{1/2} \myVM{a}{b}{} \right\rVert _{2}$ measures the alignment of  $\myVM{a}{b}{}$ with the scattered component of the direct path. Hence, when the RIS-BS aligns strongly with the direct LOS path then $\kd$ improves the SNR. Conversely, when the RIS-BS aligns strongly with the direct scattered channel then increasing $\kd$ reduces the SNR.

The dominant efect here is the $\Abs{\myVM{a}{b}{H}\ad{}}{}<\sqrt{2A}$ case, where $\kd$ decreases the mean SNR. This is because it is unlikely for two independent steering vectors to strongly align, whereas $\myVM{a}{b}{}$ can align with $\myVM{R}{d}{}$ in several ways as the correlation matrix is usually full-rank (see \cite{ISinghRayleigh}).

\subsection{Effect of $\rho_{\mathrm{ru}}$ on $\Exp{\text{SNR}}$}\label{SubSec: Effect of rhoru on E{SNR}}
The effect of correlation in the UE-RIS channel, $\rho_{\mathrm{ru}}$, on the mean SNR is confined to the variable $F_R$ in \eq{\ref{Eq: Corr E{SNR} Ricean}}. Identifying the behavior of $F_R$ with respect to $\rho_{\mathrm{ru}}$ is very difficult due to the oscillations in the series expansion given in \eqref{Eq: Whittaker M} caused by the $\cos(n\theta)$ term. 

 
However, when $\kru=0$ we have the correlated Rayleigh case and $F_R$ becomes $F$ in \eq{\ref{Eq: HyperGeometric}} which increases with correlation  (see \cite{ISinghRayleigh}). Hence, we conjecture the same broad trend here, so that the mean SNR usually benefits from increasing correlation in $\myVM{h}{ru}{}$.

It is worth noting that the majority of the existing literature does not have the complication of the $\cos(n\theta)$ term in the bivarate Ricean PDF. This is due to the assumption that both of the Ricean variables have identical LOS components (see \cite{Zogas} for example). However, this simplifying assumption is not valid here and the more complex version in \cite{Yacoub} is required.

\subsection{Effect of $\rho_{\mathrm{d}}$ on $\Exp{\text{SNR}}$}\label{SubSec: Effect of rhod on E{SNR}}
Only variable $A$ in the first term of \eq{\ref{Eq: Corr E{SNR} Ricean}} is affected by $\rho_{d}$. It is shown in \cite[Sec.~IV.B]{ISinghRayleigh} that as $\rho_{\text{d}} \rightarrow 1$, we can usually expect $A \leq \sqrt{M}$ for large $M$. In contrast, when $\rho_{\mathrm{d}} = 0$, we have $A=\sqrt{M}$. Hence, the mean SNR benefits from low correlation in the UE-BS link.

\subsection{Favourable and Unfavourable Conditions}\label{SubSec: Max, Min, Fav E{SNR}}

From Sec.~\ref{SubSec: Effect of kru on E{SNR}} - Sec.~\ref{SubSec: Effect of rhod on E{SNR}}, low correlation and low K-factor in the UE-BS link along with a high K-factor in the UE-RIS link tend to improve $\Exp{\text{SNR}}$. Hence, we define the \textit{favourable channel scenario} as an iid Rayleigh channel between UE and BS and pure LOS between UE and RIS. Conversely, the \textit{unfavourable channel scenario} comprises an iid Rayleigh UE-RIS channel and a LOS UE-BS channel. 


Using the analysis in Sec.~\ref{Sec: Perf. Insights based on E{SNR}}, the mean SNR in the favourable channel scenario is given by,
\begin{align}\label{Eq: Fav. E{SNR}}
	& \Exp{\mathrm{SNR}_{\mathrm{fav}}} \notag \\
	& = \left( \betad M + N\sqrt{M\pi}\sqrt{\betad\betabr\betaru} + \betabr\betaru M N^{2}\right) \bar{\tau},
\end{align}
while the mean SNR in the unfavourable channel scenario is,
\begin{align}\label{Eq: UnFav. E{SNR}}
& \Exp{\mathrm{SNR}_{\mathrm{unfav}}} \notag \\
& = \Bigg( \betad M + N\sqrt{\pi}\Abs{\myVM{a}{b}{H}\myVM{a}{d}{}}{}\sqrt{\betad\betabr\betaru} \notag \\
& \hspace{7em} + \betabr\betaru M\left(N + \frac{\pi N(N-1)}{4}\right)\Bigg) \bar{\tau},
\end{align}
using the result for $F_R$ for uncorrelated Rayleigh fading (see \cite[Eq.~(14)]{ISinghRayleigh}).

Next we consider the relative difference between these two scenarios for large RIS sizes. Defining the gain 
\begin{equation}
	\mathrm{Gain}_{\mathrm{fav-unfav}} 
	= \frac{\Exp{\mathrm{SNR}_{\mathrm{fav}}} - \Exp{\mathrm{SNR}_{\mathrm{unfav}}}}{\Exp{\mathrm{SNR}_{\mathrm{unfav}}}},  \label{Eq: Gain fav-unfav} 
\end{equation}
it is straightforward to show that
\begin{equation*}
\lim\limits_{N\rightarrow\infty} \mathrm{Gain}_{\mathrm{fav-unfav}} = 
	\frac{4-\pi}{\pi} \approx 27.32\%.
\end{equation*}
Hence, for large RIS, the asymptotic relative gain between  the channel scenarios is approximately  27.32\%. 

In order to find the maximum gain, we rewrite \eqref{Eq: Gain fav-unfav}   as
\begin{equation}\label{Eq: Gain fav-unfav Relaxed} 
	\mathrm{Gain}_{\mathrm{fav-unfav}} \triangleq \frac{N^2D_1 + ND_2}{N^2D_3 + ND_4 + D_5}.
\end{equation}

\noindent Differentiating \eqref{Eq: Gain fav-unfav Relaxed} with respect to $N$ is elementary and allows the maximum to be obtained. Details are omitted but results are shown in Sec.~\ref{Sec: Results}.
\section{Results}\label{Sec: Results}
We present numerical results to verify the analysis in Sec.~\ref{Sec: Perf. Insights based on E{SNR}}. Firstly, note that we do not consider cell-wide averaging as the focus is on the SNR distribution over the fast fading. Furthermore, the relationship between the SNR and the path gains, $\beta_{\mathrm{d}},\beta_{\mathrm{br}}$ and $\beta_{\mathrm{ru}}$, is straightforward, as shown in \eq{\ref{Eq: Corr E{SNR} Ricean}}. Hence, we present numerical results for fixed link gains. In particular, since the RIS-BS link is LOS, we assume $\beta_{\mathrm{br}}=d_{\mathrm{br}}^{-2}$ where $d_{\mathrm{br}}=20$m. Next, for simplicity, $\beta_{\mathrm{d}}=\beta_{\mathrm{ru}}=0.69$. This was chosen to give the 95\%-ile of the SNR distribution as 25 dB in the baseline case of moderate channel correlation and identical Ricean K-factors for the LOS and scattered paths (defined as $\rho_{\mathrm{d}}=\rho_{\mathrm{ru}}=0.7$, $\kd=\kru=1$), with $M=32$ BS antennas and $N=64$ RIS elements.

As stated in Sec. \ref{Sec: Channel Model}, the steering vectors for $\myVM{H}{br}{}$ are not restricted to any particular formation. However, for simulation purposes, we will use the VURA model as outlined in \cite{CMiller}, but in the $y-z$ plane with equal spacing in both dimensions at both the RIS and BS. The $y$ and $z$ components of the steering vector at the BS are $\myVM{a}{b,y}{} $ and $\myVM{a}{b,z}{}$ which are given by
\begin{align*}
& [1, e^{j2\pi d_{\mathrm{b}} \sin(\theta_{\mathrm{A}})\sin(\omega_{\mathrm{A}})}, \ldots, e^{j2\pi d_{\mathrm{b}} (M_{y}-1) \sin(\theta_{\mathrm{A}})\sin(\omega_{\mathrm{A}})}]^{T} \\
& \text{and } [1, e^{j2\pi d_{\mathrm{b}} \cos(\theta_{\mathrm{A}})}, \ldots, e^{j2\pi d_{\mathrm{b}}(M_{z}-1) \cos(\theta_{\mathrm{A}})}]^{T},
\end{align*}
respectively. Similarly at the RIS, $\myVM{a}{r,y}{}$ and $\myVM{a}{r,z}{}$ are defined by,
\begin{align*}
& [1, e^{j2\pi d_{\mathrm{r}} \sin(\theta_{\mathrm{D}})\sin(\omega_{\mathrm{D}})}, \ldots, e^{j2\pi d_{\mathrm{r}}(N_{y}-1) \sin(\theta_{\mathrm{D}})\sin(\omega_{\mathrm{D}})}]^{T} \\
& \text{and } [1, e^{j2\pi d_{\mathrm{r}} \cos(\theta_{\mathrm{D}})}, \ldots, e^{j2\pi d_{\mathrm{r}} (N_{z}-1) \cos(\theta_{\mathrm{D}})}]^{T},
\end{align*}
respectively, where $M = M_{y}M_{z}$, $N = N_{y}N_{z}$ with $M_{y}, M_{z}$ being the number of antenna columns and rows at the BS and $N_{y},N_{z}$ being the number of columns and rows of RIS elements, $d_{\mathrm{b}}=0.5$, $d_{\mathrm{r}}=0.2$, where $d_{\mathrm{b}}$ and $d_{\mathrm{r}}$ are in wavelength units. The steering vectors at the BS and RIS are then given by,
\begin{equation}
\myVM{a}{b}{} = \myVM{a}{b,y}{} \otimes \myVM{a}{b,z}{}
\quad,\quad
\myVM{a}{r}{} = \myVM{a}{r,y}{} \otimes \myVM{a}{r,z}{},
\end{equation} 
respectively. $\theta_{\mathrm{A}}$ and $\omega_{\mathrm{A}}$ are elevation/azimuth angles of arrival (AOAs) at the BS and $\theta_{\mathrm{D}},\omega_{\mathrm{D}}$ are the corresponding angles of departure (AODs) at the RIS. The elevation/azimuth angles are selected based on the following geometry representing a range of LOS $\myVM{H}{br}{}$ links with less elevation variation than azimuth variation: 
$	
\theta_{D} \sim \mathcal{U}[70^{o},90^{o}], \hspace{1em}
\omega_{D} \sim \mathcal{U}[-30^{o},30^{o}], \hspace{1em}  
\theta_{A} = 180^{o} - \theta_{D}, \hspace{1em} 
\omega_{A} \sim \mathcal{U}[-30^{o},30^{o}] 
$. For all results in this paper we use a single sample from this range of angles given by $\theta_{D}=77.1^{o}, \omega_{D}=19.95^{o}, \theta_{A}=109.9^{o}, \omega_{A}=-29.9^{o}$.

For the LOS components in the UE-BS and UE-RIS channels, the $y$ and $z$ components of the steering vector at the BS are $\myVM{a}{d,y}{}$ and $\myVM{a}{d,z}{}$ which are given by
\begin{align*}
& [1, e^{j2\pi d_{\mathrm{b}} \sin(\theta_{\mathrm{A,d}})\sin(\omega_{\mathrm{A,d}})}, \ldots, e^{j2\pi d_{\mathrm{b}} (M_{y}-1) \sin(\theta_{\mathrm{A,d}})\sin(\omega_{\mathrm{A,d}})}]^{T} \\
& \text{and }[1, e^{j2\pi d_{\mathrm{b}} \cos(\theta_{\mathrm{A,d}})}, \ldots, e^{j2\pi d_{\mathrm{b}}(M_{z}-1) \cos(\theta_{\mathrm{A,d}})}]^{T},
\end{align*}
respectively. Similarly at the RIS, $\myVM{a}{ru,y}{}$ and $\myVM{a}{ru,z}{}$ are defined by,
\begin{align*}
& [1, e^{j2\pi d_{\mathrm{r}} \sin(\theta_{\mathrm{D,r}})\sin(\omega_{\mathrm{D,r}})}, \ldots, e^{j2\pi d_{\mathrm{r}}(N_{y}-1) \sin(\theta_{\mathrm{D,r}})\sin(\omega_{\mathrm{D,r}})}]^{T} \\
& \text{and }[1, e^{j2\pi d_{\mathrm{r}} \cos(\theta_{\mathrm{D,r}})}, \ldots, e^{j2\pi d_{\mathrm{r}} (N_{z}-1) \cos(\theta_{\mathrm{D,r}})}]^{T},
\end{align*}
respectively. Therefore, the LOS components in the UE-BS and UE-RIS channels are given by,
\begin{equation}
\myVM{a}{d}{} = \myVM{a}{d,y}{} \otimes \myVM{a}{d,z}{}
\quad,\quad
\myVM{a}{ru}{} = \myVM{a}{ru,y}{} \otimes \myVM{a}{ru,z}{},
\end{equation} 
where $\theta_{\mathrm{A,d}}$ and $\omega_{\mathrm{A,d}}$ are elevation/azimuth angles of arrival (AOAs) for the UE at the BS and $\theta_{\mathrm{D,r}},\omega_{\mathrm{D,r}}$ are the corresponding angles of departure (AODs) for the UE at the RIS. 
For all results in this paper we use a single sample of angles from $\theta_{\mathrm{A},\mathrm{d}}, \theta_{\mathrm{D},\mathrm{r}}\sim \mathcal{U}(0,180^{\circ})$, $\omega_{\mathrm{A},\mathrm{d}}, \omega_{\mathrm{D},\mathrm{r}} \sim \mathcal{U}(-90^{\circ},90^{\circ})$ which are given by $\theta_{D,r}=80.94^{\circ}, \omega_{D,r}=-64.35^{\circ}, \theta_{A,d}=71.95^{\circ}, \omega_{A,d}=25.1^{\circ}$.

Note that all of these parameter values and variable definitions are not altered throughout the results and figures, unless specified otherwise.

\subsection{Approximate CDF for SNR}\label{SubSec: Gamma approx fit for SNR}
It is known that the SNR of a wide range of fading channels can be approximated by a mixture gamma distribution \cite{GammaApproxSNR}. Also, as discussed in Sec.~\ref{Sec: E{SNR} and Var{SNR}}, it is well-known that a single gamma approximation is often reasonable for a sum of a number of positive random variables \cite{GammaApproxY}. Motivated by this, we approximate the SNR in \eq{\ref{Eq: SNR Eq}} by a single gamma variable.

The shape parameter of a gamma approximation to the SNR is given by $k_{\gamma}=\frac{\Exp{\text{SNR}}^{2}}{\Var{\text{SNR}}}$ and the scale parameter is $\theta_{\gamma}=\frac{\Var{\text{SNR}}}{\Exp{\text{SNR}}}$, where $\Exp{\text{SNR}}$ and $\Var{\text{SNR}}$ are given in Sec.~\ref{Sec: E{SNR} and Var{SNR}}. Using these values of $k_{\gamma},\theta_{\gamma}$, the analytical and simulated SNR CDFs  are shown in Fig.~\ref{Fig: CDF Agreements} for $N=16$ and $N=64$, both with  $\rho_{\mathrm{ru}}=\rho_{\mathrm{d}} \in \{0,0.7,0.95\}$ and $\kd=\kru\in\{1,10^3\}$. When computing the analytical SNR CDFs, for $\rho_{\mathrm{ru}}=\rho_{\mathrm{d}}=0$, \eq{\ref{Eq: Uncorr E{SNR} Ricean}} and \eq{\ref{Eq: Corr Var{SNR} Ricean_uncorr}} were used but with $F_R,C_{u1},C_{u2}$ presented in Sec.~\ref{SubSec: Special Case: UnCorr Ricean}. For $\rho_{\mathrm{ru}}=\rho_{\mathrm{d}}=0.7$, \eq{\ref{Eq: Corr E{SNR} Ricean}} and \eq{\ref{Eq: Corr Var{SNR} Ricean}} were used. For $\rho_{\mathrm{ru}}=\rho_{\mathrm{d}}=0.95$, \eq{\ref{Eq: Corr E{SNR} Ricean}} and \eq{\ref{Eq: Corr Var{SNR} Ricean}} were also used but with $F_R$ given by \eq{\ref{Eq: FR for high kru and rhoru}}.
\begin{figure}[h]
	\centering
	\includegraphics[width=\columnwidth]{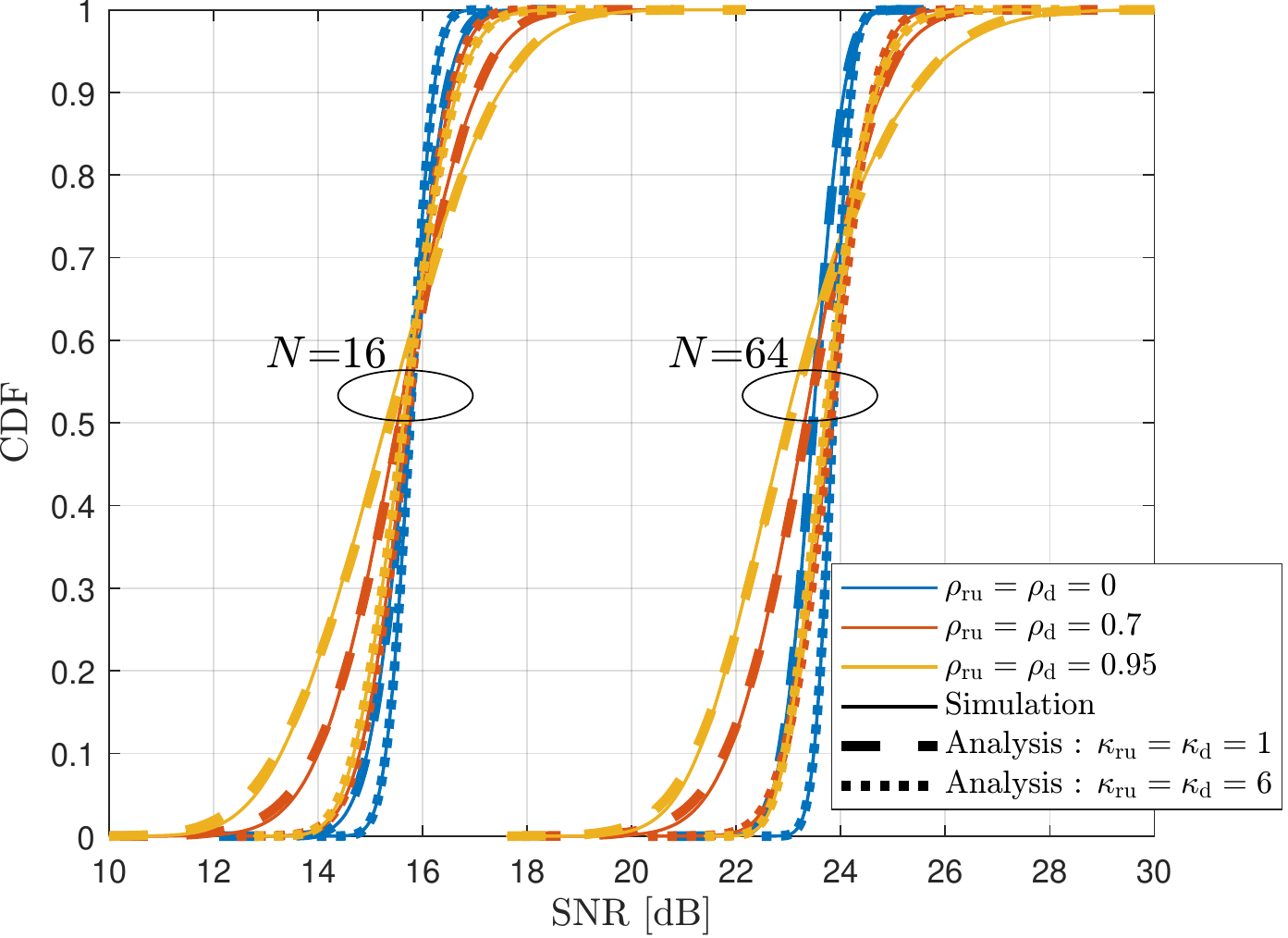}
	\raisecapt\caption{Simulated and analytical CDFs for $N\in\{16,64\}$, both with  $\rho_{\mathrm{ru}}=\rho_{\mathrm{d}}\in\{0,0.7,0.95\}$ and $\kd=\kru\in\{1,6\}$}
	\label{Fig: CDF Agreements}
\end{figure}

As expected, there is a very good agreement between the simulated and analytical SNR CDFs when $\rho_{\mathrm{ru}}=\rho_{\mathrm{d}}=0$ due to exact mean SNR and SNR variance expressions being available for all $\kd$ and $\kru$. For higher correlation, we can see that that the agreement deviates slightly in the lower and upper tails when the K-factor is small. This is due to the approximation made by fitting a gamma distribution using the 3rd and 4th moments of $Y$. For higher K-factors, the 3rd and 4th moments of $Y$ are better approximated due to the UE-RIS channel being dominated by the LOS path, which will cause $Y$ to be close to deterministic. In such scenarios, the gamma distribution provides an even better representation of the UL SNR. In general, for all correlation and K-factor scenarios, the gamma distribution provides a good representation of the UL SNR even in high correlation scenarios.

\subsection{The effects of $\rho_{\mathrm{d}},\rho_{\mathrm{ru}},\kd,\kru$ and Asymptotic Results}\label{SubSec: Asymptotic Analysis Results}
Here, we verify the performance insights based on $\Exp{\mathrm{SNR}}$ described in Sec.~\ref{Sec: Perf. Insights based on E{SNR}}. Fig.~\ref{Fig: Mean SNR analysis} gives the SNR simulations and analyses for various correlation and K-factor combinations in $\myVM{h}{d}{}$ and $\myVM{h}{ru}{}$. As a general observation, it is apparent that the analysis agrees with the simulations for every scenario. 

Inspection of Fig.~\ref{Fig: Mean SNR analysis} reveals that increasing $\kd$ has a negative impact on the mean SNR as predicted by the analysis. The analysis also predicted that correlation in the direct channel $\rho_{\mathrm{d}}$ negatively impacts performance, which is also seen in Fig.~\ref{Fig: Mean SNR analysis}. The best performance occurs for $\rho_{\mathrm{d}}=0,\rho_{\mathrm{ru}}=1,\kd=1,\kru=10^3$, which supports the favourable channel scenario claim in the analyses. Finally, the worst performance is observed when $\rho_{\mathrm{d}}=\rho_{\mathrm{ru}}=0,\kd=10^3,\kru=1$, supporting the unfavourable channel scenario claim in Sec.~\ref{Sec: Perf. Insights based on E{SNR}}. 

Furthermore, although not visible, scenarios where the channel has a very high K-factor will yield very similar curves regardless of the correlation in that channel. For example, the curve for $\rho_{\mathrm{d}}=0,\rho_{\mathrm{ru}}=1,\kd=10^3,\kru=1$ is very similar to the curve for $\rho_{\mathrm{d}}=1,\rho_{\mathrm{ru}}=1,\kd=10^3,\kru=1$. This observation is not unusual since the channel is dominated by the LOS path and is thus not significantly impacted by scattering. 

\begin{figure}[h]
	\centering
	\includegraphics[width=\columnwidth]{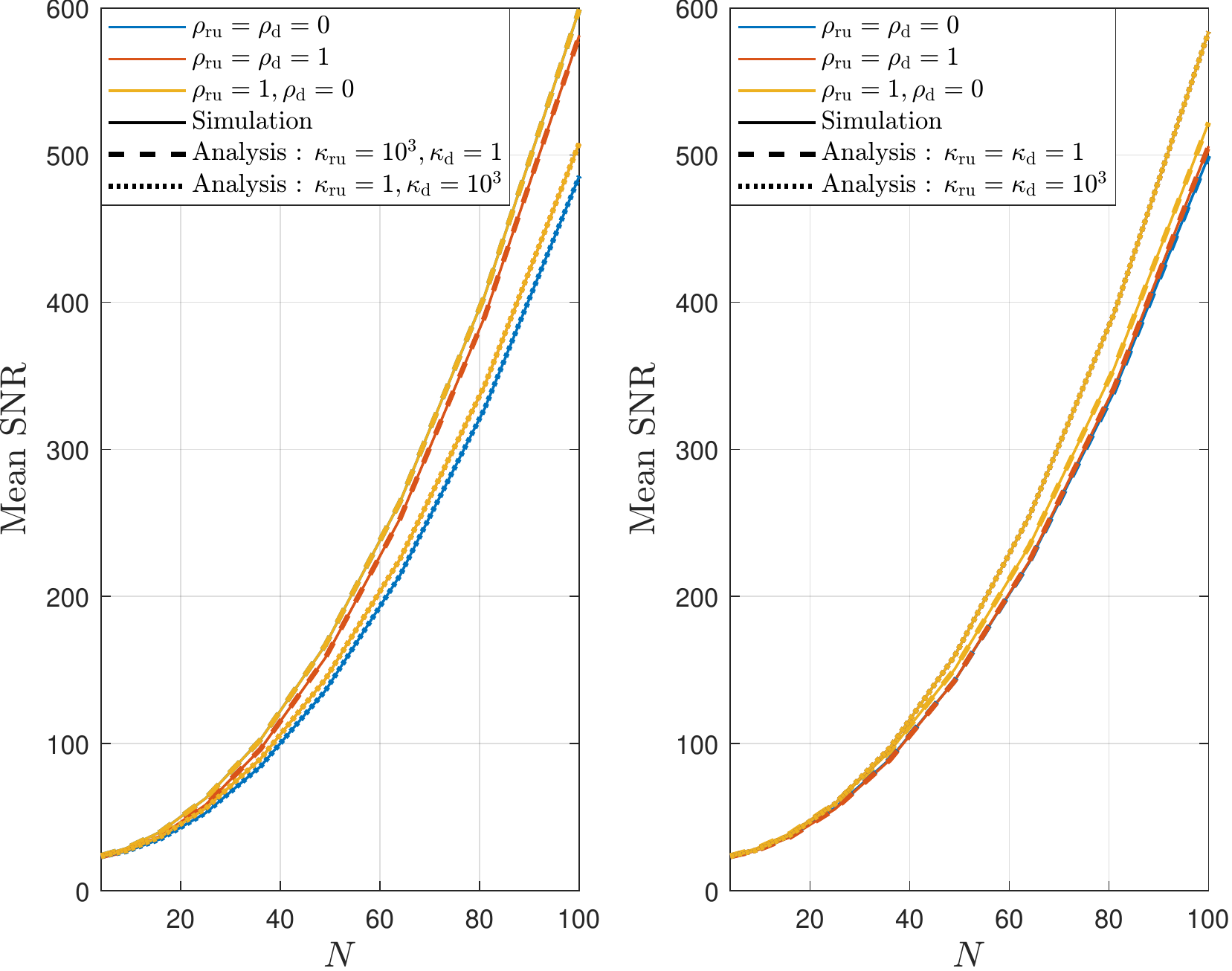}
	\raisecapt\caption{Simulated and analytical mean results for SNR for correlation scenarios: $\rho_{\mathrm{ru}}=\rho_{\mathrm{d}}\in\{0,1\},\rho_{\mathrm{ru}}=1,\rho_{\mathrm{d}}=0$ and K-factor scenarios: $\kd\in\{1,10^3\},\kru\in\{1,10^3\}$.}
	\label{Fig: Mean SNR analysis}
\end{figure}
Fig.~\ref{Fig: Var SNR Analysis} shows the accuracy of the SNR variance approximation. As expected, for scenarios where $\myVM{h}{d}{}$ and $\myVM{h}{ru}{}$ are uncorrelated, we note perfect agreement between the simulations and analyses. As the K-factor in the UE-RIS channel, $\kru$, increases, the variance is reduced, which becomes significant when both $\kru$ and $\kd$ are very large. This effect is also observed by the change in the CDF spread in Fig.~\ref{Fig: CDF Agreements}. 

For scenarios where the K-factor is low, the accuracy of the agreement between simulation and analyses is a consequence of the approximations of the 3rd and 4th moments of $Y$ only, since all of the other terms in the expression can be expressed in closed form. From Fig.~\ref{Fig: Var SNR Analysis} it can be seen that as the number of RIS elements $N$ increases, the approximation worsens for highly correlated scenarios with $\kru=1$ and $\kd=10^3$. However, for the identical correlation scenarios with $\kru=10^3$ and $\kd=1$, the approximation improves. Note that $Y$ is the sum of amplitudes of the correlated elements in $\myVM{h}{ru}{}$. When the K-factor is very high, then $\Exp{Y^3} \approx Y^3$ and $\Exp{Y^4}\approx Y^4$, hence the observed improvement in the approximation.

The convergence of the various curves in Fig.~\ref{Fig: Var SNR Analysis} supports the claim in the analyses that increasing the number of RIS elements reduces the negative impacts of $\kd$ and $\rho_{\mathrm{d}}$ in the UE-BS channel. This is seen in the curves relating to $\rho_{\mathrm{ru}}=\rho_{\mathrm{d}}=1$ and $\rho_{\mathrm{ru}}=1,\rho_{\mathrm{d}}=0$ as well as the curves relating to $\kru=\kd=1$ and $\kru=1,\kd=10^3$.
\begin{figure}[h]
	\centering
	\includegraphics[width=\columnwidth]{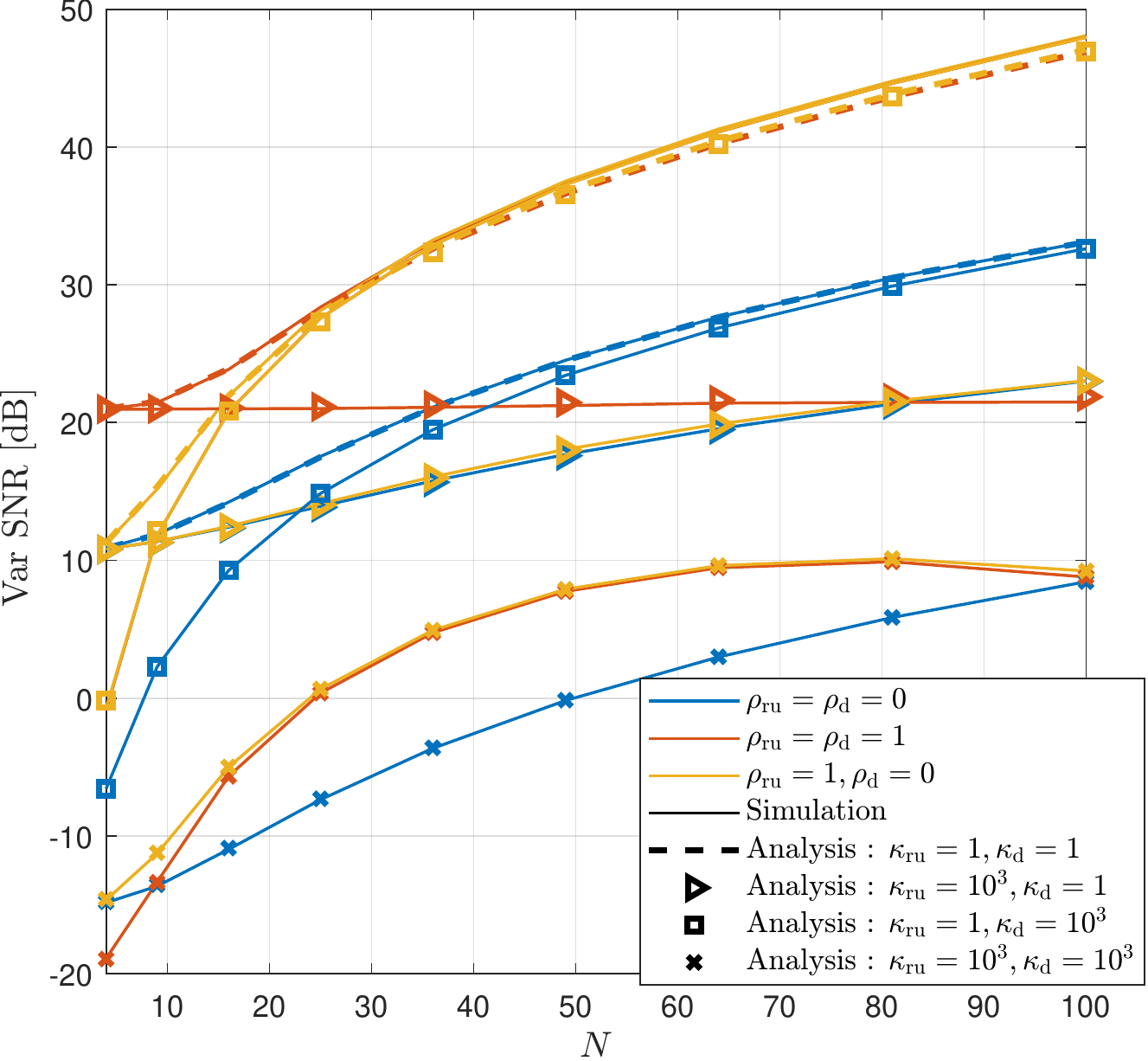}
	\raisecapt\caption{Simulated and analytical variance results for SNR for correlation scenarios: $\rho_{\mathrm{ru}}=\rho_{\mathrm{d}}\in\{0,1\},\rho_{\mathrm{ru}}=1,\rho_{\mathrm{d}}=0$ and K-factor scenarios: $\kd\in\{1,10^3\},\kru\in\{1,10^3\}$.}
	\label{Fig: Var SNR Analysis}
\end{figure}

Finally, we verify the results regarding the relative gain given in Sec.~\ref{Sec: Perf. Insights based on E{SNR}}. The relative gain improvement in \eq{\ref{Eq: Gain fav-unfav}}, when transitioning from unfavourable to favourable channel scenarios, is simulated for a range of different link gains $\betad,\betabr,\betaru$ and shown in Fig.~\ref{Fig: Optimization and Asymptotic Analysis}. 
\begin{figure}[h]
	\centering
	\includegraphics[width=\columnwidth]{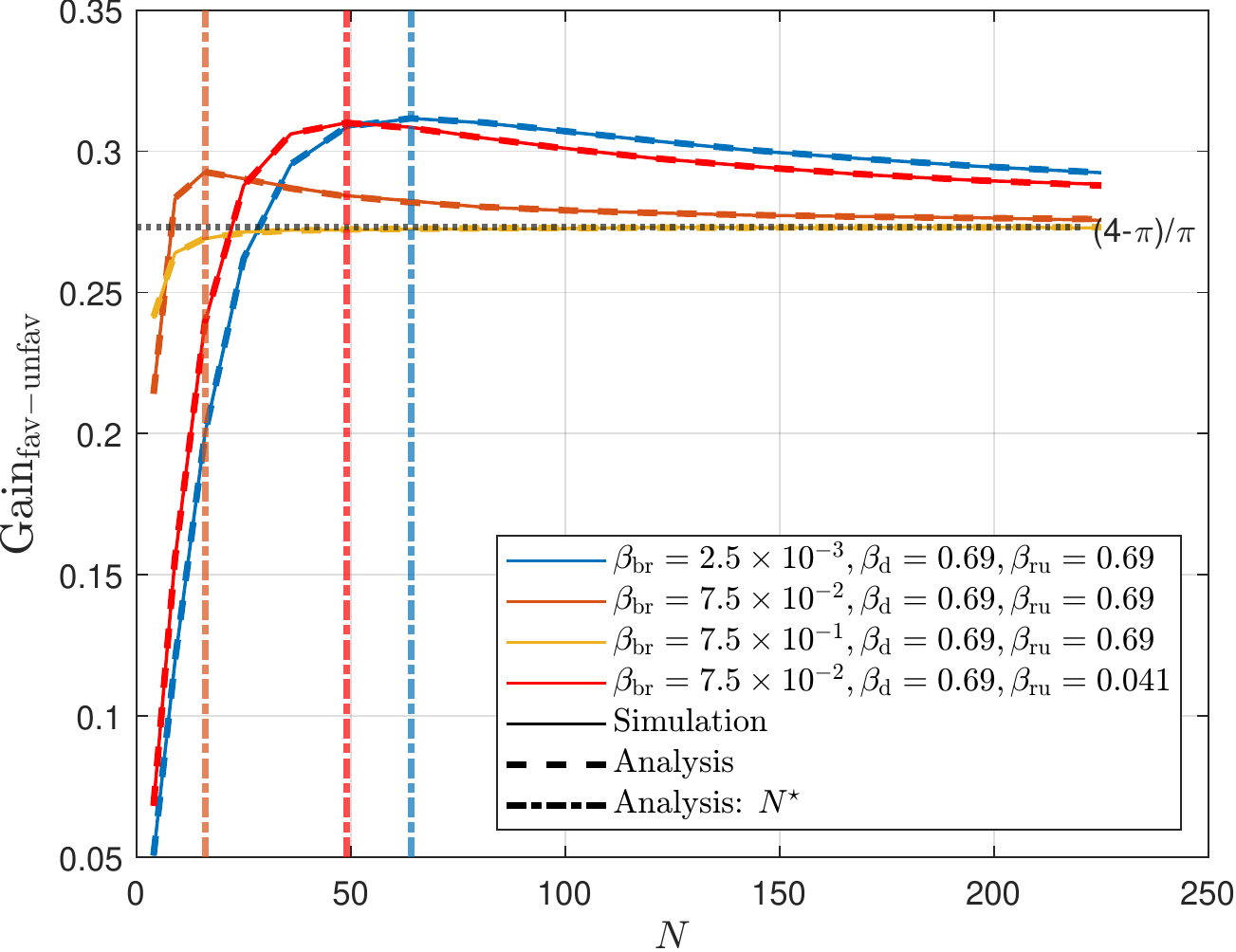}
	\raisecapt\caption{Simulated and analytical relative gain results for a range of link gains $\betad,\betabr,\betaru$.}
	\label{Fig: Optimization and Asymptotic Analysis}
\end{figure}
As can be seen, all of the different link gain scenarios asymptotically approach $(4-\pi)/\pi$ as per the analysis and the saturation rate is dependent on the link gain values. Also apparent is that having a large number of RIS elements does not always yield the maximum relative gain. The vertical lines represent the number of RIS elements required to achieve the maximum possible gain, as detailed in Sec.~\ref{SubSec: Max, Min, Fav E{SNR}}. In general, the simulation and analysis agree.

\section{Conclusion}\label{Sec: Conclusion}
We derive an exact closed form expression for the mean SNR of the optimal single user RIS design where spatially correlated Ricean fading is assumed for the UE-BS and UE-RIS channels and the RIS-BS channel is LOS. We also provide an accurate approximation of the SNR variance and a gamma approximation to the CDF of the SNR. Furthermore, we show that when the UE-RIS and UE-BS channels are correlated Rayleigh, both the mean SNR and SNR variance expressions under correlated Ricean channels can be reduced to the expressions given in \cite{ISinghRayleigh}. The analysis presented offers new insight into how spatial correlation and the Ricean K-factor impact the mean SNR and scenarios in which we would expect high and low SNR performance. Specifically, we show that the mean SNR benefits from having a high correlation and a high K-factor in UE-RIS channel while the UE-BS channel has low correlation and a low K-factor. These correlation and K-factor extremes represent favorable channel scenarios while unfavorable channel scenarios occur at the opposite ends of these extremes. To this end, we show that the asymptotic gain achievable when transitioning between unfavourable and favourable environments is $(4-\pi)/\pi$. 
 

\begin{appendices}
	
\section{Derivation of mean SNR}\label{AppA: E{SNR} Corr Deri}
Expanding \eq{\ref{Eq: SNR Eq}} gives,
\begin{align*}
	\text{SNR} & = \left( \betad\hdtil{H}\hdtil{} + 2\sqrt{\betad}\Re\{\alpha\hdtil{H}\myVM{a}{b}{}\} + \Abs{\alpha}{2}M \right) \bar{\tau} \\
	& \triangleq (S_{1} + S_{2} + S_{3}) \bar{\tau}.
\end{align*}
We compute $\Exp{\text{SNR}}$ by computing the expectation of each term in the expression above.

\vspace{1em}
\noindent \textbf{Term 1}: Expanding $S_{1}$,
\begin{equation*}
	S_{1} = \betad\left( 
	\etad^{2}\ad{H}\ad{} +  2\etad\zetad\Re\{ \ad{H}\hdNLOS{} \} 
	+ \zetad^{2}\myVM{u}{d}{H}\myVM{R}{d}{}\myVM{u}{d}{}
	\right).
\end{equation*}
Using \cite[Eq. (9)]{8422818} and since $\myVM{u}{d}{} \sim \mathcal{CN}(\mathbf{0},\mathbf{I})$, then
\begin{equation}\label{Eq: S1}
	\Exp{S_{1}} = \betad\left(\etad^{2}M + \zetad^{2}\tr{R}{d}{} \right) = \betad M.
\end{equation}

\vspace{1em}
\noindent \textbf{Term 2}: Substituting $\alpha$ from Sec. \ref{Sec: Optimal SNR} into $S_{2}$,
\begin{align*}
	\Exp{S_2} 
	& =2 \sqrt{\beta_{\mathrm{d}}\beta_{\mathrm{br}}\beta_{\mathrm{ru}}} 
	\mathbb{E}\left\{ \sum_{n=1}^{N}\Abs{\myVMIndex{\tilde{h}}{ru}{n}{}}{} \right\} \mathbb{E}\left\{ \Abs{\myVM{a}{b}{H} \hdtil{}}{} \right\} \\
	& = 2\sqrt{\beta_{\mathrm{d}}\beta_{\mathrm{br}}\beta_{\mathrm{ru}}} 
	\mathbb{E}\left\{ Y \right\} 
	\mathbb{E}\left\{ \Abs{\myVM{a}{b}{H} \hdtil{}}{} \right\}.
\end{align*}
Firstly, we can rewrite the definition of the $i^{th}$ element of $\hrutil{}$ in Sec. \ref{Sec: Channel Model} as $\mathbf{\tilde{h}}_{\mathrm{ru},i} = X_{1} + jX_{2}$ where $X_{1} \sim \mathcal{N}(\etaru\Re\{\aru{}\},(\zetaru^{2}/2)\mathbf{I})$ and $X_{2} \sim \mathcal{N}(\etaru\Im\{\aru{}\},(\zetaru^{2}/2)\mathbf{I})$. Then using the moments of a Ricean RV \cite[Eq. (2.3-58)]{DigCom} and a confluent hypergeometric to Laguerre polynomial relation \cite[Eq. (13.6.9)]{Stegun}, we have
\begin{equation*}
	\mathbb{E}\left\{ Y \right\}  = \frac{N\zetaru\sqrt{\pi}}{2} \LagP{-\kru}.
\end{equation*}
Similarly,
\begin{equation*}
	\Exp{\Abs{\myVM{a}{b}{H} \hdtil{}}{}} = 
	\frac{A \zetad \sqrt{\pi}}{2} \LagP{-\frac{\kd \Abs{\myVM{a}{b}{H}\ad{}}{2}}{A^{2}}},
\end{equation*}
where $A = \Norm{\myVM{R}{d}{1/2}\myVM{a}{b}{}}{2}{}$.	Therefore,
\begin{align}\label{Eq: S2}
	& \Exp{S_{2}} \notag \\
	& = \frac{NA\zetad\zetaru\pi\sqrt{\betad\betabr\betaru}}{2}
	\LagP{-\kru}\LagP{-\frac{\kd \Abs{\myVM{a}{b}{H}\ad{}}{2}}{A^{2}}}.
\end{align}

\vspace{1em}
\noindent \textbf{Term 3}: Since  $\Abs{\psi}{} = 1$ (where $\psi$ is given in Sec.~\ref{Sec: Optimal SNR}) we have $S_3= 
M\beta_{\mathrm{br}}\beta_{\mathrm{ru}} Y^{2}$ and expanding $Y$ gives
\begin{align*}
	\mathbb{E}\left\{ Y^{2} \right\} 
	= N + \underset{i \neq k}{\sum_{i=1}^{N} \sum_{k=1}^{N}} \mathbb{E}\left\{ \Abs{\myVMIndex{\tilde{h}}{ru}{i}{}}{} \Abs{\myVMIndex{\tilde{h}}{ru}{k}{}}{}  \right\}.
\end{align*}
Let $r_{i}=\Abs{\myVMIndex{\tilde{h}}{ru}{i}{}}{}$, $r_{k}=\Abs{\myVMIndex{\tilde{h}}{ru}{k}{}}{}$, $a = \etaru$, $\sigma_{\mathrm{ru}}^{2} = \zetaru^{2}/2$, $\rho_{ik} = \left(\myVM{R}{ru}{}\right)_{i,k}$, then the joint moment of two correlated Rician RVs $r_{i},r_{k}$ is given by \cite[Eq. (27)]{Yacoub},
\begin{align*}
	& \Exp{r_{i}r_{k}} = \frac{\left(1 - \Abs{\rho_{ik}}{2}\right)^2}{1 + \kru} 
	\mathrm{exp}\left\{ -\frac{2\kru - 2\mu_c\kru}{1 - \Abs{\rho_{ik}}{2}}\right\} \\
	& \times \sum_{m=0}^{\infty}\sum_{n=0}^{m} \frac{\epsilon_n \cos(n\phi) \Abs{\rho_{ik}}{2m - n}}{m! (m-n)! (n!)^2} \left( \frac{\kru(1 + \Abs{\rho_{ik}}{2} - 2\mu_c)}{1 - \Abs{\rho_{ik}}{2}} \right)^n \\
	& \times \Gamma^2\left( m + \frac{3}{2}\right) 
	{}_{1}F_{1}^2\left(m+\frac{3}{2}, n+1, \frac{\kru(1 + \Abs{\rho_{ij}}{2} - 2\mu_c)}{1 - \Abs{\rho_{ik}}{2}}  \right),
\end{align*}
where ${}_{1}F_{1}(\cdot)$ is the confluent hypergeometric function, $\phi = \arg\{ (1 + \Abs{\rho_{ik}}{2})\mu_c\kru - 2\kru\Abs{\rho_{ik}}{2} + j(1 - \Abs{\rho_{ik}}{2}\mu_s\kru)\}$, $\mu_c = \Abs{\rho_{ik}}{}\cos(\Delta\theta)$, $\mu_s = \Abs{\rho_{ik}}{}\sin(\Delta\theta)$, $\Delta\theta = \arg\{\mathbf{\aru{}}_{,i}\} - \arg\{\mathbf{\aru{}}_{,k}\}$ and $\epsilon_n$  is given in Theorem 1 (see \cite[Eq.~(26b)]{Yacoub}).

Finally, we substitute the final form of $\Exp{r_{i}r_{k}}$ into the double summation of $\Exp{Y^{2}}$ and denote this as $F_{\mathbf{R}}$. Therefore,
\begin{equation}\label{Eq: S3}
	\Exp{S_{2}} = M\betabr\betaru(N + F_{\mathrm{R}}),
\end{equation}
where $F_{\mathrm{R}}$ is given by 	\eq{\ref{Eq: Whittaker M}}.
Combining \eq{\ref{Eq: S1}}, \eq{\ref{Eq: S2}} and \eq{\ref{Eq: S3}} completes the derivation.

\section{Derivation of SNR Variance}\label{App: Var{SNR} Corr Deri}
To compute the variance we take the square of \eq{\ref{Eq: SNR Eq}} giving,
\begin{align}\label{Eq: SNR^2}
\text{SNR}^{2} & = 
\Big( \betad^{2} \left( \hdtil{H}\hdtil{} \right)^{2} + 
4\betad^{3/2} \hdtil{H}\hdtil{} \Re\left\{\alpha \hdtil{H}\myVM{a}{b}{}  \right\} \notag \\
& +  2\betad M \Abs{\alpha}{2} \hdtil{H}\hdtil{}  + 
4\betad\Re\left\{\alpha \hdtil{H}\myVM{a}{b}{}  \right\}^{2} \notag \\
& + 4 \sqrt{\beta_{\mathrm{d}}} M \Abs{\alpha}{2} \Re\left\{ \alpha\hdtil{H}\myVM{a}{b}{} \right\} + \Abs{\alpha}{4} M^{2}  \Big) \bar{\tau}^{2}\notag \\
& \triangleq (T_1+T_2+T_3+T_4+T_5+T_6)\bar{\tau}^{2}.
\end{align}

\vspace{1em}
\noindent \textbf{Term 1}: $T_{1}$ can be alternately written as $T_{1} = \betad \Norm{\hdtil{}}{2}{4}$. Using \cite[Eq. (5)]{7976368}, we have
\begin{equation}
	\Exp{T_{1}} = \betad^{2} \left(
	M^{2}  + 
	2\etad^{2}\zetad^{2} \Norm{\myVM{R}{d}{1/2}\ad{}}{2}{2} +
	\zetad^{4}\tr{R}{d}{2}
	\right).
\end{equation}

\vspace{1em}
\noindent \textbf{Term 2}: Substituting $\alpha$ from Sec. \ref{Sec: Optimal SNR} into $T_{2}$,
\begin{align}\label{App: E{T2}}
	\Exp{T_{2}} & = 4 \betad^{3/2} \sqrt{\betabr\betaru} \Exp{Y} \Exp{\hdtil{H}\hdtil{}\Abs{\hdtil{H}\myVM{a}{b}{}}{}} \notag \\
	& \mkern-36mu = 2 \zetaru \betad^{3/2}\sqrt{\betabr\betaru} N \sqrt{\pi} L_{1/2}(-\kru) \Exp{\hdtil{H}\hdtil{}\Abs{\hdtil{H}\myVM{a}{b}{}}{}},
\end{align} 
where the expectation of $Y$ is found in App.~\ref{AppA: E{SNR} Corr Deri}. To compute $\Exp{\hdtil{H}\hdtil{}\Abs{\hdtil{H}\myVM{a}{b}{}}{}}$, we introduce the following variables: $\myVM{Q}{}{} = \myVM{P}{}{H}\myVM{R}{d}{}\myVM{P}{}{}$ and $\myVM{p}{1}{} = \myVM{R}{d}{1/2}\myVM{a}{b}{}/A$
where $\myVM{p}{1}{}$ is the first column of any arbitrary orthonormal matrix $\myVM{P}{}{}$, $A = \Norm{\myVM{R}{d}{1/2}\myVM{a}{b}{}}{2}{}$, $\myVM{x}{}{} \sim \mathcal{CN}(\mathbf{0},\mathbf{I})$, $b = \etad\ad{H}\myVM{a}{b}{}/A$ and $\myVM{d}{}{} = \etad \myVM{P}{}{H} \myVM{R}{d}{-1/2} \ad{} $. Then $\Exp{\hdtil{H}\hdtil{}\Abs{\hdtil{H}\myVM{a}{b}{}}{}}$ can be rewritten as,
\begin{align}\label{App: E{Complicating}}
	&\Exp{\hdtil{H}\hdtil{}\Abs{\hdtil{H}\myVM{a}{b}{}}{}} \notag \\ 
	& = A \Bigg(\sum_{s=1}^{M}\sum_{t=1}^{M} \zetad^{2}\myVM{Q}{st}{}\Exp{ x_{s}^{H}x_{t}\Abs{\zetad x_{1} + b}{}} \notag \\
	& + 2\Re\left\{ \sum_{s=1}^{M}\sum_{t=1}^{M} \zetad\myVM{Q}{st}{}d_{s}^{H}\Exp{ x_{t}\Abs{\zetad x_{1} + b}{}}\right\} \notag \\
	& + \sum_{s=1}^{M}\sum_{t=1}^{M} \myVM{Q}{st}{}d_{s}^{H}d_{t}\Exp{\Abs{\zetad x_{1} + b}{}} \Bigg) \notag \\
	& \triangleq A \left(T_{21} + T_{22} + T_{23}\right).
\end{align}
Computing \eq{\ref{App: E{Complicating}}} and substituting it into \eq{\ref{App: E{T2}}} completes the derivation for $\Exp{T_2}$.

Sub-Term 1: Notice that $T_{21}$ is non zero iff $s=t$. So,
\begin{align*}
	T_{21} & = \myVM{Q}{11}{} \Exp{\Abs{\zetad x_{1}}{2}\Abs{\zetad x_{1} + b}{}} \\
	& \qquad + \Exp{\Abs{\zetad x_{1} + b}{}} \left( \tr{Q}{}{} - \myVM{Q}{11}{} \right) \\
	& = B\zetad^{2}\mathbb{E}\Big\{ \Abs{\zetad x_{1} + b}{3} - 2\Re\{ b^{H}(\zetad x_{1} + b)\Abs{\zetad x_{1} + b}{} \}  \\
	& \qquad + \Abs{b}{2}\Abs{\zetad x_{1} + b}{} \Big\} + \Exp{\Abs{\zetad x_{1} + b}{}}\left( M - B \right),
\end{align*}
since $\tr{Q}{}{} = M$ and $B = \mathbf{Q}_{11} = \Norm{\myVM{R}{d}{}\myVM{a}{b}{}}{2}{2}/A^2$. The closed form solution for $\mathcal{I} \triangleq \Exp{(\zetad x_{1} + b)\Abs{\zetad x_{1} + b}{} \}}$ is found in App.~\ref{App: Curly I} 
\footnote{It is worth noting that App.~\ref{App: Curly I} also gives the general solution for $\Exp{(ax + b)\Abs{ax + b}{}}$ where $x\sim\mathcal{CN}(0,1)$, $a\in\mathbb{R}$ and $b\in\mathbb{C}$.}.
Using this result, along with the 1st and 3rd moments of a Ricean RV and some algebraic simplifiction, 
\begin{multline}\label{Eq: T21}
		T_{21} = \frac{3B\zetad^{3}\sqrt{\pi}}{4} L_{3/2}\left(-C\right)  
		- \frac{2B\etad}{A}\Re\{\myVM{a}{b}{H}\ad{} \mathcal{I} \} \\
		+ \frac{\zetad}{2}L_{1/2}\left(-C\right) \left[ \frac{B\kd\Abs{\ad{H}\myVM{a}{b}{}}{2}}{A^2 (1+\kd)} + \frac{\sqrt{\pi}(M-B)}{1+\kd}\right],
\end{multline}
where $C = \kd\Abs{\ad{H}\myVM{a}{b}{}}{2}/A^{2}$.

Sub-Term 2: The terms inside the real operator in $T_{22}$ can be written as,
\begin{align*}
	\sum_{s=1}^{M} \myVM{Q}{s1}{}\myVM{d}{s}{H} \left( \Exp{(\zetad x_{1} + b)\Abs{\zetad x_{1} + b}{}} - b\Exp{\Abs{\zetad x_{1}}{b}} \right).
\end{align*}
Using App.~\ref{App: Curly I} and the 1st moment of a Ricean RV,
\begin{align}
	& T_{22} \notag \\ 
	& = \frac{2\etad}{A}
	\Re\left\{
	\ad{H}\myVM{R}{d}{}\myVM{a}{b}{} \left[ \mathcal{I} - \frac{b\zetad\sqrt{\pi}}{2}L_{1/2}\left(-\frac{\kd\Abs{\ad{H}\myVM{a}{b}{}}{2}}{A^{2}} \right) \right] \right\}.
\end{align}

Sub-Term 3: Using the first moment of a Ricean RV and noting that $\myVM{d}{}{H}\myVM{Q}{}{}\myVM{d}{}{} = \etad^2 M$,
\begin{equation}
	T_{23} = \frac{\etad^{2}\zetad M \sqrt{\pi}}{2}L_{1/2}\left(-\frac{\kd\Abs{\ad{H}\myVM{a}{b}{}}{2}}{A^{2}} \right).
\end{equation}

The result of $\Exp{T_2}$ is,
\begin{align}
	\Exp{T_{2}} & = 
	2 \zetaru \betad^{3/2}\sqrt{\betabr\betaru} N \sqrt{\pi} L_{1/2}(-\kru) \notag \\
	& \qquad \times A (T_{21} + T_{22} + T_{23}).
\end{align}

\vspace{1em}
\noindent \textbf{Term 3}: Substituting $\alpha$ from Sec. \ref{Sec: Optimal SNR} into $T_{3}$,
\begin{align}
	\Exp{T_{3}} & = 2M\betad\betabr\betaru\Exp{Y^{2}}\Exp{\hdtil{H}\hdtil{}} \notag \\
	& = 2M^{2}\betad\betabr\betaru(N + F_{R}),
\end{align}
where the results for $\Exp{\hdtil{H}\hdtil{}}$ and $\Exp{Y^{2}}$ are given in App.~\ref{AppA: E{SNR} Corr Deri} and $F_{\mathrm{R}}$ is given by \eq{\ref{Eq: Whittaker M}}.

\vspace{1em}
\noindent \textbf{Term 4}: Substituting $\alpha$ from Sec. \ref{Sec: Optimal SNR} into $T_{4}$,
\begin{align}
	& \Exp{T_{4}} = 4\betad\betabr\betaru\Exp{Y^{2}}\Exp{\Abs{\myVM{a}{b}{H}\hdtil{}}{2}} \\
	& = 4\betad\betabr\betaru(N+F_{\mathrm{R}})\Exp{\Abs{\myVM{a}{b}{H}\hdtil{}}{2}},
\end{align}
using the result for $\Exp{Y^{2}}$ in App.~\ref{AppA: E{SNR} Corr Deri}. Notice that $\Exp{\Abs{\myVM{a}{b}{H}\hdtil{}}{2}}$ is the second moment of $\Abs{\myVM{a}{b}{H}\hdtil{}}{}$. Using the result for $\Exp{\Abs{\myVM{a}{b}{H}\hdtil{}}{}}$ in App.~\ref{AppA: E{SNR} Corr Deri} and the moments of a Ricean RV \cite[Eq. (2.3-58)]{DigCom},
\begin{equation*}
	\Exp{\Abs{\myVM{a}{b}{H}\hdtil{}}{2}} = 
	\etad^{2}\Abs{\myVM{a}{b}{H}\ad{}}{2} + \zetad^{2}\Norm{\myVM{R}{d}{1/2}\myVM{a}{b}{}}{2}{2}.
\end{equation*}
Therefore,
\begin{equation}
	\Exp{T_{4}} = 4\betad\betabr\betaru(N+F_{\mathrm{R}})\left(
	\etad^{2}\Abs{\myVM{a}{b}{H}\ad{}}{2} + \zetad^{2}A^{2}
	\right),
\end{equation}
where $A$ is given in App.~\ref{AppA: E{SNR} Corr Deri} and $F_{\mathrm{R}}$ is given by \eq{\ref{Eq: Whittaker M}}.

\vspace{1em}
\noindent \textbf{Term 5}: Substituting $\alpha$ from Sec. \ref{Sec: Optimal SNR} into $T_{5}$,
\begin{align*}
	& \Exp{T_{5}} = 4M\sqrt{\betad}\left( \betabr\betaru \right)^{3/2} 
	\Exp{Y^{3}} \Exp{\Abs{\myVM{a}{b}{H}\hdtil{}}{}} \notag \\
	& = 2AM\zetad\sqrt{\pi\betad}\left( \betabr\betaru \right)^{3/2} 
	  \LagP{-\frac{\kd \Abs{\myVM{a}{b}{H}\ad{}}{2}}{A^{2}}}
	\Exp{Y^{3}},
\end{align*}
using the result for $\Exp{\Abs{\myVM{a}{b}{H}\hdtil{}}{}}$ in App.~\ref{AppA: E{SNR} Corr Deri}. The variable $Y^{3}$ is a sum of products of the magnitudes of three correlated Ricean random variables. To the best of our knowledge the mean of such terms is intractable without the use of multiple infinite summations and special functions. As such we will use an approximation based on the gamma distribution to approximate $\Exp{Y^3}$ (see App.~\ref{App: E{Y^3},E{Y^4} Approx}). Substituting this approximation, we have
\begin{multline}
	\Exp{T_{5}} \\
	\approx 
	2AM\zetad\sqrt{\pi\betad}\left( \betabr\betaru \right)^{3/2} 
	\LagP{-\frac{\kd \Abs{\myVM{a}{b}{H}\ad{}}{2}}{A^{2}}}
	C_{1},
\end{multline}
with $C_{1} = b^3 a \prod_{k=1}^{2}(k+a)$ where $a$ and $b$ are defined in App.~\ref{App: E{Y^3},E{Y^4} Approx}.

\vspace{1em}
\noindent \textbf{Term 6}: Substituting $\alpha$ from Sec. \ref{Sec: Optimal SNR} into $T_{6}$,
\begin{equation*}
	\Exp{T_{6}} = \left( M\betabr\betaru \right)^{2} \Exp{Y^{4}}.
\end{equation*}
Since $Y^4$ is even more complex than $Y^3$, we re-use the gamma approximation for $Y$ to make progress.  Using the approximation for $\Exp{Y^{4}}$ in App.~\ref{App: E{Y^3},E{Y^4} Approx} gives,
\begin{equation}
	\Exp{T_{6}} \approx \left( M\betabr\betaru \right)^{2} C_{2},
\end{equation}
with $C_{2} = b^4 a \prod_{k=1}^{3}(k+a)$ where $a$ and $b$ are defined in App.~\ref{App: E{Y^3},E{Y^4} Approx}.

\section{$\Exp{(ax+b)\Abs{ax+b}{}}$ and $\mathcal{I}$}\label{App: Curly I}
Here, we derive the expected value of $(ax+b)\Abs{ax+b}{}$ where $x \sim \mathcal{CN}(0,1)$, $a \in \mathbb{R}$ and $b \in \mathbb{C}$. Let $ax+b = \rho e^{j\theta}$ where $\rho$ and $\theta$ are the amplitude and phase of $ax+b$ respectively. Then, the joint PDF of the amplitude and phase is \cite[Eq. (2.4)]{Miller}
\begin{equation*}
	f(\rho,\theta) = 
	\frac{\rho}{\pi a^2} 
	\exp\left\{-\frac{\rho^2}{a^2} + 2\frac{\Abs{b}{}}{a^2}\rho \cos(\theta - \angle b) -\frac{\Abs{b}{2}}{a^2} \right\}.
\end{equation*}
Using the joint PDF, the expectation $\Exp{(ax+b)\Abs{ax+b}{}}$ is,
\begin{align}\label{Eq: Gen E{Z|Z|}}
	&\Exp{\rho^2 e^{j\theta}} \notag \\
	& = \frac{e^{j\angle b - \frac{\Abs{b}{2}}{a^2}}}{\pi a^2}
	\int_{0}^{\infty} \rho^3 e^{-\frac{\rho^2}{a^2}}
	\int_{0}^{2\pi} e^{j\theta}e^{2\frac{\Abs{b}{}}{a^2}\rho \cos(\theta)}
	d\theta d\rho \notag \\
	& \overset{(a)}{=} \frac{2e^{j\angle b - \frac{\Abs{b}{2}}{a^2}}}{a^2}
	\int_{0}^{\infty} 
	\rho^3 e^{-\frac{\rho^2}{a^2}} I_{1}\left( \frac{2\Abs{b}{}\rho}{a^2} \right)
	d\rho \notag \\
	& \overset{(b)}{=}
	\frac{-3 a^3 \sqrt{\pi}}{4 \Abs{b}{}} 
	\exp\left\{ j\angle b - \frac{\Abs{b}{2}}{2a^2} \right\}
	M_{3/2,1/2}\left( -\frac{\Abs{b}{2}}{a^2} \right),
\end{align}
where $M_{3/2,1/2}(\cdot)$ is the Whittaker M function. $(a)$ uses \cite[Eq. (3.937.1)]{GradRyz} to evaluate the integral with respect to $\theta$ while utilizing the fact that the integral over one period of an even function multiplied by an odd function is zero. $(b)$ uses \cite[Eq. (6.631.1)]{GradRyz} to evaluate the integral with respect to $\rho$ along with some algebraic simplifications. 

The SNR variance expression in \eq{\ref{Eq: Corr Var{SNR} Ricean}} utilizes this result with specific values for $a$ and $b$; i.e., $a = \zetad$ and $b = \etad \ad{H}\myVM{a}{b}{}/A$ where $A$ can be found in App.~\ref{AppA: E{SNR} Corr Deri}. Substituting these values of $a$ and $b$ into \eq{\ref{Eq: Gen E{Z|Z|}}} and after some simplification we have,
\begin{align}\label{Eq: Curly I}
	\mathcal{I} & = 
	\frac{-3 A \sqrt{\pi}}{4\sqrt{\kd}(1+\kd)} 
	\exp\left\{ j\angle \ad{H}\myVM{a}{b}{} - \frac{\kd\Abs{\ad{H}\myVM{a}{b}{}}{2}}{2 A^2} \right\} \notag \\
	& \times M_{3/2,1/2}\left( -\frac{\kd\Abs{\ad{H}\myVM{a}{b}{}}{2}}{A^2} \right),
\end{align}
where $\mathcal{I}$ in \eq{\ref{Eq: Fancy I}} is required in Theorem 2.

\section{Approximations for $\Exp{Y^3}$ and $\Exp{Y^4}$}\label{App: E{Y^3},E{Y^4} Approx}
As discussed in Sec.~\ref{Sec: E{SNR} and Var{SNR}}, we employ a gamma approximation for $Y$. From App.~\ref{AppA: E{SNR} Corr Deri}, we know that $\Exp{Y}=\frac{N\zetaru\sqrt{\pi}}{2} \LagP{-\kru}$ and  $\Exp{Y^2} = N+F_{\mathrm{R}}$, where $F_{\mathrm{R}}$ is defined by \eq{\ref{Eq: Whittaker M}}. Then, the variance of $Y$ is $\Var{Y}= N + F_{\mathrm{R}} - \frac{N^2 \zetaru^2 \pi^2}{4} L_{1/2}^{2}(-\kru)$. Using the method of moments, the parameters that define a gamma distribution fit for $Y$  are, 
\begin{align*}
	a & = \dfrac{\Exp{Y}^{2}}{\Var{Y}} = \dfrac{N^{2} \pi \zetaru^{2} L_{1/2}^{2}(\kru)}{4(N+F_{\mathrm{R}}) - N^{2} \pi \zetaru^{2} L_{1/2}^{2}(\kru)}, \\
	b & = \dfrac{\Var{Y}}{\Exp{Y}} \\
	& = \dfrac{2}{N \sqrt{\pi} \zetaru L_{1/2}(\kru)}\left( N + F_{\mathrm{R}} - \frac{N^2 \zetaru^2 \pi^2}{4} L_{1/2}^{2}(-\kru) \right),
\end{align*}
where $a$ and $b$ are the shape and scale parameters respectively. Suppose $X \sim \mathcal{G}(a,b)$, then  the 3$^{\text{rd}}$ and $4^{\text{th}}$ moments are 
\begin{align*}
	\Exp{X^3}  & = b^3 a \prod_{k=1}^{2}(k+a), \quad
	\Exp{X^4}  = b^4 a \prod_{k=1}^{3}(k+a).
\end{align*}
Substituting $a$ and $b$ into the above moments yields the approximations for $\Exp{Y^3}$ and $\Exp{Y^4}$.

\section{Correlated Rayleigh equivalent of $F_{\mathrm{R}}$}\label{App: Reduce FR to Corr Rayleigh}
When $\myVM{h}{ru}{}$ is a Rayleigh fading channel, we have $\kru=0$, $\phi=0$ and $I_{0}(0) = 1$. The integral form of $\Exp{r_{i}r_{k}}$ given by \cite[Eq. (42)]{Yacoub} then becomes
\begin{align*}
	\Exp{r_{i}r_{k}} & = \frac{4}{1 - \rhoruij{2}} \int_{0}^{\infty} \int_{0}^{\infty} r_i^2 r_k^2 e^{-(r_i^2 + r_k^2)/ (1 - \rhoruij{2})} \\
	& \hspace{6em} \times I_{0}\left( \frac{ 2\rhoruij{} r_i r_k}{1-\rhoruij{2}} \right) \ dr_i dr_k.
\end{align*}
Replacing the modified Bessel function of the first kind with its infinite series equivalent, we have
\begin{align}\label{App: Integral for of FR in Rayleigh case}
	& \Exp{r_{i}r_{k}}\\
	& = \sum_{m=0}^{\infty} \frac{4\rhoruij{2m}}{(1 - \rhoruij{2})^{2m+1} (m!)^2}
	\vspace{3em} \left[\int_{0}^{\infty} x^{2m+2}e^{-\frac{x^2}{(1-\rhoruij{2})}} \ dx \right]^2
\end{align}
Using \cite[Eq.~(3.461.2)]{GradRyz} and \cite[Eq.~(6.1.12)]{Stegun}, the integral in \eq{\ref{App: Integral for of FR in Rayleigh case}} is evaluated to be,
\begin{align*}
	\int_{0}^{\infty} x^{2m+2}e^{-\frac{x^2}{(1-\rhoruij{2})}} \ dx 
	& = \frac{(3/2)_m}{4}\sqrt{\pi}(1-\rhoruij{2})^{m + 3/2},
\end{align*}
where $(\cdot)_m$ is the Pochhammer symbol. Squaring and substituting into $\Exp{r_i r_k}$ yields, 
\begin{align*}	
	\Exp{r_{i}r_{k}} & \overset{(a)}{=} \frac{\pi}{4} (1-\rhoruij{2})^{2} \sum_{m=0}^{\infty} 
	\frac{(3/2)_{m} (3/2)_{m}}{(1)_{m}} \frac{\rhoruij{2m}}{m!} \\
	& \overset{(b)}{=}  \frac{\pi}{4} (1-\rhoruij{2})^{2}
	{}_{2}F_{1}\left(\frac{3}{2},\frac{3}{2};1;\left\lvert\rho_{ik}\right\rvert^2 \right),
\end{align*}
where $(a)$ uses the fact that $(1)_{m} = m!$ and $(b)$ replaces the infinite series with its Gaussian hypergeometric equivalent.  Taking the summation of $\Exp{r_i, r_k}$ over all RIS elements $N$, we get
$$
F = \underset{i \neq k}{\sum_{i=1}^{N} \sum_{k=1}^{N}} \dfrac{\pi}{4}\left( 1 - \left\lvert\rho_{ik}\right\rvert^2\right)^2 {}_{2}F_{1}\left(\frac{3}{2},\frac{3}{2};1;\left\lvert\rho_{ik}\right\rvert^2 \right),
$$
Which is identical to the correlated Rayleigh form of $F_{\mathrm{R}}$ in \cite[Eq.~(9)]{ISinghRayleigh}.
\vspace{-1em}
\section{Integral form for $\Exp{r_i r_k}$ with $\rhoruij{}=1$}\label{App: Integral form of E{ri rj}}
For high correlation and K-factor values, the expectation of correlated Ricean random variables in \eq{\ref{Eq: Whittaker M}} is computationally expensive to compute. In order to obtain an accurate result in such circumstances, the number of terms in the double summation becomes very large. Here, we propose the use of numerical integration to compute such expectations.

Let $\myVM{h}{sc}{} = \myVM{R}{ru}{1/2} \myVM{u}{ru}{}$ denote the normalized value of the scattered component of the UR-RIS channel in Sec.~\ref{Sec: Channel Model}. Then, 
\begin{align*}
	\Exp{r_i r_k} &= 
	\Exp{\Abs{\etaru \myVMIndex{a}{ru}{i}{} + \zetaru\myVMIndex{h}{sc}{i}{}}{} \Abs{\etaru \myVMIndex{a}{ru}{k}{} + \zetaru\myVMIndex{h}{sc}{k}{}}{} }.
\end{align*}
Now, let $\myVMIndex{h}{sc}{k}{} = \rho_{ik}\myVMIndex{h}{sc}{i}{} + e\sqrt{1 - \rhoruij{2}}$ where $e \sim \mathcal{CN}(0,1)$ and let $\rho_{ik} = \rhoruij{}e^{j\phi}$, $\myVMIndex{\bar{a}}{ru}{i}{} = \myVMIndex{a}{ru}{i}{}e^{j\phi}$, then
\begin{align*}
	\Exp{r_i r_k} &= 
	\mathbb{E}\Bigg\{ \Abs{\etaru\myVMIndex{\bar{a}}{ru}{i}{} + \zetaru\myVMIndex{\bar{h}}{sc}{i}{}}{} \\
	& \times \Abs{\etaru\myVMIndex{\bar{a}}{ru}{k}{} + \zetaru\left(\rhoruij{}\myVMIndex{\bar{h}}{sc}{i}{} + e\sqrt{1-\rhoruij{2}}\right)}{} \Bigg\},
\end{align*}
where $\myVMIndex{\bar{h}}{sc}{i}{} \sim \mathcal{CN}(0,1)$. Taking the expectation over $e$ and noting that this term gives the mean of a Ricean random variable, we have
\begin{align*}
	\Exp{r_i r_k} &=
	\frac{\sqrt{\pi} \zetaru \sqrt{1 - \rhoruij{2}}}{2} 
	\mathbb{E}\Bigg\{\Abs{\etaru\myVMIndex{\bar{a}}{ru}{i}{} + \zetaru\myVMIndex{\bar{h}}{sc}{i}{}}{} \\
	& \times \LagP{-\frac{\Abs{\etaru\myVMIndex{\bar{a}}{ru}{k}{} + \zetaru\rhoruij{}\myVMIndex{\bar{h}}{sc}{k}{}}{2}}{\zetaru^2 (1-\rhoruij{2})}}\Bigg\}.
\end{align*}
Finally, let $a_i = \etaru \myVMIndex{\bar{a}}{ru}{i}{}/\zetaru = a_{iR} + ja_{iI}$, $a_k = \etaru\myVMIndex{a}{ru}{k}{}/(\zetaru\rhoruij{}) = a_{jR} + ja_{jI}$, $b = \rhoruij{2}/(1-\rhoruij{2})$. Then we have,
\begingroup
\allowdisplaybreaks
\begin{align*}
	&\Exp{r_i r_j} \\ 
	&= \frac{\sqrt{\pi} \zetaru^2 \sqrt{1 - \rhoruij{2}}}{2} \mathbb{E}\Big\{\sqrt{(a_{iR} + X)^2 + (a_{iY} + Y)^2} \\ 
	& \hspace{3em} \times \LagP{-b((a_{jR} + X)^2 + (a_{jY} + Y)^2)} \Big\} \\
	& = \frac{\zetaru^2 \sqrt{1 - \rhoruij{2}}}{2\sqrt{\pi}} 
	\int_{-\infty}^{\infty} \int_{-\infty}^{\infty} \sqrt{(a_{iR} + X)^2 + (a_{iY} + Y)^2} \\
	& \times \LagP{-b((a_{jR} + X)^2 + (a_{jY} + Y)^2)} e^{-(x^2 + y^2)} \ dydx \\
	& \overset{(a)}{=} \frac{\zetaru^2 \sqrt{1 - \rhoruij{2}}}{2\sqrt{\pi}} \\
	& \hspace{1em} \times \int_{0}^{2\pi}\int_{0}^{\infty} r e^{-r^2} \sqrt{\Abs{a_i}{2} + r^2 + 2r\Abs{a_i}{}\cos(\theta - \theta_i)} \\
	& \hspace{1em} \times \LagP{-b(\Abs{a_k}{2} + r^2 + 2r\Abs{a_k}{}\cos(\theta - \theta_k))} \ d\theta dr,
\end{align*} 
\endgroup
where $(a)$ involved transforming the integrals from cartesian to polar form, $\theta_i = \angle \myVMIndex{a}{ru}{i}{}$ and $\theta_k = \angle \myVMIndex{a}{ru}{k}{}$. Computing these integrals over $1 \leq i,k \leq N$ for $i \neq k$ yields \eq{\ref{Eq: FR for high kru and rhoru}}. It is worth noting that the integral form of $\Exp{r_i r_k}$ is computationally more efficient than the one presented in \cite[Eq.~(42)]{Yacoub}, where $\Exp{r_i r_k}$ is given as an infinite summation of double integrals, both requiring integration over the region $[0,\infty)$.

For the benchmark case, correlation $\rhoruij{} = 1$, both the double summation \eq{\ref{Eq: Whittaker M}} and integral form of $\Exp{r_i r_k}$ are invalid since substituting $\rhoruij{} = 1$ results in an indeterminate answer. However, we can use the integral form and find its result as $\rhoruij{} \rightarrow 1$. Firstly notice that correlation is only present in the term outside the integrals and also inside the Laguerre function. Let $z = \Abs{a_k}{2} + r^2 + 2r\Abs{a_k}{}\cos(\theta - \theta_k)$ , then
\begin{align*}
	&\lim\limits_{\rhoruij{} \rightarrow 1} \sqrt{1 - \rhoruij{2}} \LagP{-\frac{\rhoruij{2}}{1 - \rhoruij{2}} z} \\
	&\overset{(a)}{=} \lim\limits_{\rhoruij{} \rightarrow 1} \sqrt{1 - \rhoruij{2}} {}_{1}F_{1}\left(-\frac{1}{2};1; -\frac{\rhoruij{2}}{1 - \rhoruij{2}} z\right) \\
	&\overset{(b)}{=} \lim\limits_{\rhoruij{} \rightarrow 1}
	\frac{\sqrt{\rhoruij{2}z}}{\Gamma(3/2)} \left[1 + \mathcal{O}\left(\frac{1-\rhoruij{2}}{\rhoruij{2} z}\right) \right] \\
	& = 2\sqrt{\frac{\Abs{a_k}{2} + r^2 + 2r\Abs{a_k}{}\cos(\theta - \theta_k)}{\pi}},
\end{align*}
where $(a)$ uses \cite[Eq.~(13.6.9)]{Stegun} to transform the Laguerre function into its confluent hypergeometric equivalent and $(b)$ uses \cite[Eq.~(13.1.5)]{Stegun} to replace the confluent hypergeometric function with its asymptotic equivalent. Substituting this limit into the integral yields,
\begin{align*}
	\Exp{r_i r_k} &= \frac{\zetaru^2}{\pi} \int_{0}^{2\pi} \int_{0}^{\infty} r e^{-r^2} \\
	& \hspace{1em} \times \sqrt{\Abs{a_i}{2} + r^2 + 2r\Abs{a_i}{}\cos(\theta - \theta_i)} \\
	& \hspace{1em} \times \sqrt{\Abs{a_i}{2} + r^2 + 2r\Abs{a_i}{}\cos(\theta - \theta_k)} \ d\theta dr,
\end{align*}
where we make the substitution of $a_k = a_i$ at $\rhoruij{} = 1$. Computing these integrals over $1 \leq i,k \leq N$ for $i \neq k$ yields \eq{\ref{Eq: FR for rhoru = 1}}.

\end{appendices}

\bibliographystyle{IEEEtran}
\bibliography{RIS_Paper_Journal}

\end{document}